\documentclass[11pt,reqno,notitlepage]{article}
\usepackage{setspace}
\usepackage{amssymb}
\usepackage{amsthm}
\usepackage{graphicx}
\usepackage{amsmath}
\usepackage{fancyhdr}
\usepackage{datetime}
\usepackage[toc]{appendix}
\usepackage[usenames,dvipsnames]{color}
\usepackage{tikz}
\usepackage{mathrsfs}
\usepackage[normalem]{ulem}
\usepackage{graphicx, xcolor}

\usepackage{enumerate}

\usepackage{setspace}
\usepackage{graphicx, xcolor}												
\usepackage[mathscr]{eucal}												
\usepackage{subcaption}
\usepackage{color}
\usepackage{epstopdf}

\allowdisplaybreaks

\usetikzlibrary{arrows,decorations.pathreplacing}
\usepackage[top=1in, bottom=1in, left=1in, right=1in]{geometry}
\usepackage{mathtools}									%
\mathtoolsset{showonlyrefs=true}							%

\newtheorem{theorem}{Theorem}
\newtheorem{lemma}[theorem]{Lemma}								%
\newtheorem{proposition}[theorem]{Proposition}	
\newtheorem{corollary}[theorem]{Corollary}	
\newtheorem{assumption}[theorem]{Assumption}	

\newtheorem{remark}{Remark}

\numberwithin{equation}{section}	
\numberwithin{theorem}{section}
\newcommand{\blue}[1]{{\color{blue}{{#1}}}}
\renewcommand{\blue}[1]{{#1}}

\renewcommand{\vec}[1]{\mathbf{#1}}

\def\E{\mathbb{E}}			
\def\P{\mathbb{P}}										
\def\EP{\E}
\def\Q{\mathbb{Q}}
\def\R{\mathbb{R}}
\def\X{\mathcal{X}}
\def\Rplus{\mathbb{R}_{++}}

\def\Zhat{\hat{Z}}
\def\bbd{\mathbb{D}}
\def\D{\bbd}

\def\ncal{\mathcal{N}}

\def\eps{\varepsilon}
\def\om{\omega}

\def \Y {\mathcal{Y}}

\def\d{\partial}		

\def\ind{\mathbb{I}}

\def \abs#1{\left| #1 \right| }
\newcommand{\e}[1]{\operatorname{e}^{#1}}
\newcommand{\cl}{\operatorname{cl}}

\newcommand{\dom}{\operatorname{dom}}

\def\define{:=}

\renewcommand{\vec}[1]{\mathbf{#1}}	
\def\f{\vec{f}}
\def\s{\vec{s}}

\def\q{\vec{q}}

\def\T{\top}
\DeclareMathOperator*{\argmax}{arg\,max}

\DeclareMathOperator*{\esssup}{ess\,sup}
\DeclareMathOperator*{\essinf}{ess\,inf}
\DeclareMathOperator*{\FIX}{FIX}

\newtheorem{example}[theorem]{Example}

\long\def\symbolfootnote[#1]#2{\begingroup\def\thefootnote{\fnsymbol{footnote}}\footnote[#1]{#2}\endgroup}

\begin{document}

\title{Endogenous inverse demand functions} 

\author{
Maxim Bichuch
\thanks{
Department of Applied Mathematics and Statistics,
Johns Hopkins University
3400 North Charles Street, 
Baltimore, MD 21218. 
{\tt mbichuch@jhu.edu}. Work  is partially supported by NSF grant DMS-1736414. Research is partially supported by the Acheson J. Duncan Fund for the Advancement of Research in Statistics.}
\and  Zachary Feinstein
\thanks{
School of Business,
Stevens Institute of Technology,
Hoboken, NJ 07030, USA,
{\tt  zfeinste@stevens.edu}. }
}
\date{\today}
\maketitle

\begin{abstract}
In this work we present an equilibrium formulation for price impacts. This is motivated by the B\"uhlmann equilibrium in which assets are sold into a system of market participants, e.g.\ a fire sale in systemic risk, and can be viewed as a generalization of the Esscher premium.  Existence and uniqueness of clearing prices for the liquidation of a portfolio are studied.  We also investigate other desired portfolio properties including monotonicity and concavity. 
Price per portfolio unit sold is also calculated.  In special cases, we study price impacts generated by market participants who follow the exponential utility and power utility.
\end{abstract}


\section{Introduction}\label{sec:intro}
Buying or selling assets in a financial market impact the prices upward or downward.  Quantifying these price impacts \blue{is} fundamental to many problems within finance, e.g., optimal liquidation and systemic risk.  
\blue{In the price-mediated contagion literature (see, e.g.,~\cite{CFS05,GLT15,AFM16,braouezec2017strategic}),} these price impacts from liquidating assets are modeled by an ``inverse demand function'' which maps the number of assets sold into market prices.  \blue{The ``demand'' nomenclature in the inverse demand function denotes the broader market demand for any liquidated assets as the broader market must take the other side of these transactions.}
In \blue{the price-mediated contagion literature referenced above}, the inverse demand function is typically chosen to follow simple analytical forms for tractability rather than for some financial meaning.  Two classical inverse demand functions -- the linear \cite{GLT15} and exponential \cite{CFS05} -- do indeed have simple financial interpretation (constant absolute and relative price impacts respectively)\blue{, however, these interpretations are provided without any economic justification for why these types of price impacts hold.}  

\blue{
In contrast to these exogenous forms, the goal of this work is to find the fair price of liquidated assets (from, e.g., fire sales) when sold into a market of agents to endogenize the price impacts.
We construct this equilibrium market as a generalization of the risk sharing economies of \cite{arrow1954existence,buhlmann1980economic,buhlmann1984general}.  Hereafter we refer to the equilibrium risk sharing frameworks of \cite{arrow1954existence,buhlmann1980economic,buhlmann1984general} as the B\"uhlmann equilibrium setting, which is formally presented in Section~\ref{sec:buhlmann}.  Briefly, we consider a market of utility maximizing agents. Each agent starts with some (risky) endowment and can trade with other agents so that a market clearing condition (i.e., conservation of total market risk) is satisfied.  
An equilibrium, if it exists and is achieved, provides a set of (asset) transfers as well as a probability measure which is used to price any claim.
In contrast to prior works, we modify the market clearing condition so that the market participants need to, on net, purchase any liquidated claim.  In this way, our modification of the B\"uhlmann equilibrium setting constructs an equilibrium pricing measure allowing us to define a fair price for the liquidated claim.
}

\blue{
In this modified B\"uhlmann equilibrium setting, we investigate the properties of the resulting pricing function.  First, we study the existence and uniqueness of the fair price of a liquidated claim.  Under suitable assumptions, further properties can be placed on this pricing function such as continuity, monotonicity, and decreasing marginal returns.
}
In addition, we investigate the resulting inverse demand function\blue{s (i.e., the mappings of the number of shares of assets or a portfolio being liquidated into a price per unit)} and derive \blue{their} properties.  \blue{These inverse demand functions are the result of an equilibrium, and 
{are thus constructed} \emph{endogenously}} in contrast to the \emph{exogenous} inverse demand functions prevalent in the systemic risk literature when studying price-mediate\blue{d} contagion.

As will be provided in Section~\ref{sec:cs}, the \blue{two aforementioned} classical inverse demand functions \blue{(linear and exponential)} can in fact be obtained in our equilibrium setting.  Importantly, the equilibrium setting considered herein relates the form of the inverse demand function to the underlying assumptions of the state of the market and returns of the traded asset(s).
\blue{In fact,} not every inverse demand function can be obtained from a given market setting; \blue{this is briefly discussed with an example in the introduction of Section~\ref{sec:cs}.} Therefore, special attention needs to be taken when exogenous forms are assumed for the inverse demand function.
Furthermore, when considering fire sales of multiple illiquid assets, it is often assumed that the liquidation of one asset does not impact prices of the other assets (see, e.g., \cite{GLT15,CS17,cont2019monitoring}).  Herein, we find that the equilibrium price impacts may not generally satisfy such a condition. \blue{We construct an example within Section~\ref{sec:cs-power}, where these inverse demand functions generate price cross-impacts even for statistically independent assets.}

\blue{Before continuing to the main body of this work, we wish to highlight two closely related fields of the literature for pricing claims in an equilibrium setting.
First,} \cite{greenwald1991transactional,bernardo2004liquidity} consider trading between two types of agents in order to derive the equilibrium trading price; those works consider market makers who value assets in a similar manner to the utility indifference price \blue{(see, e.g.,~\cite{hodges1989optimal,carmona2008indifference})}.  
\blue{Second,} \cite{gollier1996optimum,dana2007optimal,jouini2008optimal} consider a risk sharing problem between an insured agent and an insurer to determine the equilibrium premium payment.  \blue{These} works find a Pareto optimal transfer of risk and the price of that trade between two agents -- the insured and insurer.  
\blue{We wish to remind the reader that, in contrast to the two agent Pareto transfer problem of~\cite{gollier1996optimum,dana2007optimal,jouini2008optimal},} the B\"uhlmann equilibrium setting considered herein looks for a Nash equilibrium with an arbitrary number of agents.

The organization of this paper is as follows.  First, in Section~\ref{sec:buhlmann}, we will introduce the B\"uhlmann equilibrium problem and how we modify that problem in order to present the general financial setting which we will utilize throughout this work.  The main results are presented in Section~\ref{sec:main}; these results include necessary and sufficient conditions for the existence and uniqueness of clearing prices for asset liquidations.  We also find sufficient conditions for, e.g., the monotonicity and concavity of the value obtained from liquidations.  In Section~\ref{sec:cs}, we demonstrate the form and properties of our endogenous pricing functions under two special cases: when all market participants maximize the exponential utility function and when they all maximize the power utility. The proofs for all results are provided in an Online Appendix.

\section{B\"uhlmann equilibrium setup}\label{sec:buhlmann}
We are motivated in our study by the notion of the B\"uhlmann equilibrium \cite{buhlmann1980economic,buhlmann1984general} over a probability space $(\Omega,\mathcal{F},\P)$. 
Such an equilibrium endogenizes the price impacts of market behavior in a system of $n$ market participants. Each participant $1\le i\le n$, is endowed with a twice differentiable, strictly increasing and concave utility function $u_i$, \blue{with $u_i''>0$} and initial endowments $X_i \in L^\infty$ (such that $X_i \in \bbd$ a.s.) of risky payoffs. 
\blue{In this paper, we will consider two classical settings, when the utility functions are defined on the half line and on the entire real line, i.e., $\dom u_i=\Rplus := (0,\infty)$ and $\dom u_i=\R$ respectively. Throughout this work we will denote this domain as $\bbd$, i.e., $\bbd := \dom u_i$.  Additionally, we will continue to use the notation that $\Rplus := (0,\infty)$ to be the strictly positive real line while we will denote $\R_+ := [0,\infty)$ to be the nonnegative real line.
In case $\D = \Rplus$ we extend $u_i(0) = \lim\limits_{x\searrow0} u_i(x)$ in the broad sense (i.e., allowing for the possibility of $-\infty$) and $u_i(x) = -\infty$ for $x<0$.} 
Recall that the  absolute risk aversion of agent $i$ is $\rho_i(z) := -u_i''(z)/u_i'(z)\blue{>0},~z\in\bbd$. 
Throughout this work, we will denote the expectation under $\P$ as $\E := \E^{\P}$.

\begin{assumption}\label{ass:X}
\blue{Assume $X_i \in L^\infty$ such that $X_i \in \D$ a.s.\ for every market participant $i = 1,...,n$.}  Let $\X = \sum_{i = 1}^n X_i \blue{\in L^\infty},~X_i\in \D$ a.s.\ \blue{and assume} that $\essinf\X\in\bbd$.
\end{assumption}
For simplicity of exposition, we also assume a zero risk-free rate $r=0$. \blue{We use the same static setting as in B\"uhlmann \cite{buhlmann1980economic,buhlmann1984general}, where an equilibrium is solved at the initial time, and the only other time considered is}  
some future terminal time at which all randomness is resolved.  
Each market participant is assumed to be a rational agent insofar as each market participant wishes to maximize her expected utility.  
More specifically, given that each agent is endowed with (risky) endowment $X_i,~1\le i\le n$, she may choose to trade quantities $Y_i$ to reduce her risk and maximize her utility. 
For an equilibrium, the market must clear and every agent must be ``happy'' with the trade. As such, the goal is to find the clearing prices of these trades $Y_i,~1\le 1\le n$. In other words, to find the pricing measure $\Q$. This defines the solution to the B\"uhlmann equilibrium problem as a pair $(Y,\Q)$ satisfying:
\begin{enumerate}
\item \textbf{Utility maximizing}: $Y_i \in \argmax\limits_{\hat Y_i \in L^\infty}  \left\{\EP\left[u_i\left(X_i + \hat Y_i - \E^\Q[\hat Y_i]\right)\right]\right\}$ \blue{with $\EP\left[u_i\left(X_i + Y_i - \E^\Q[Y_i]\right)\right] \in \R$} for every $i \in \{1,2,...,n\}$; and
\item \textbf{Equilibrium transfers}: $\sum_{i = 1}^n Y_i = 0$.
\end{enumerate}
The measure $\Q$ is the endogenously defined probability measure which provides the price of the claims, i.e., the value of $Y_i$ at time $0$ is $\E^\Q[Y_i]$.  Notably, the pricing measure $\Q$ will (generally) differ from $\P$. 
\begin{remark}\label{rem:arrow}
We wish to note that, in the static and finite probability space setting, the B\"uhlmann equilibrium coincides with the Arrow-Debreu equilibrium~\cite{arrow1954existence} (see, e.g.,~\cite{anthropelos2017equilibrium}). We choose to utilize the B\"uhlmann equilibrium setup as it readily allows for general probability spaces without the need to consider an infinite number of commodities (and thus reducing mathematical technicalities). 
This is related to a Nash equilibrium in a pure exchange economy (see, e.g.,~\cite{duffie2010dynamic}).  While that equilibrium setup is very similar to the B\"uhlmann equilibrium problem, herein we explicitly consider the representative agent (see~\eqref{eq:rho} below); this allows us to investigate the properties of the resulting inverse demand function through a modification of the market clearing condition presented below.
\end{remark}
\begin{theorem}\label{thm:buhlmann}
There exists a unique B\"uhlmann equilibrium if:
\begin{enumerate}
\item $\bbd = \R$ and the absolute risk aversions $z \mapsto \rho_i(z) > 0$ are Lipschitz continuous, $i=1, ...,n$; or 
\item $\bbd = \Rplus$, the Inada conditions are satisfied (i.e., $\lim_{z \to 0} u_i'(z) = \infty$ and $\lim_{z \to \infty} u_i'(z) = 0$), and $z \mapsto z u_i'(z)~i=1, ...,n$ are nondecreasing.
\end{enumerate}
\end{theorem}
\begin{proof}
This is proven in \cite{buhlmann1984general} if $\bbd = \R$ and in \cite{aase1993equilibrium} if $\bbd = \Rplus$.
\end{proof}
We recall additional detail of the construction of  B\"uhlmann equilibrium in Appendix~\ref{sec:buhlmann-appendix}.

Consider the same market of $n$ participants with utility functions $u_i$ and endowments $X_i$, but now with some \emph{external} portfolio $Z \in L^\infty$ being liquidated into the market.  That is, we consider the modified B\"uhlmann equilibrium problem of determining the pair $(Y,\Q)$ satisfying:
\begin{enumerate}
\item \textbf{Utility maximizing}: $Y_i \in \argmax_{\hat Y_i} \EP\left[u_i\left(X_i + \hat Y_i - \E^\Q[\hat Y_i]\right)\right]$ \blue{with $\EP\left[u_i\left(X_i + Y_i - \E^\Q[Y_i]\right)\right] \in \R$} for every $i \in \{1,2,...,n\}$; and
\item \textbf{Equilibrium transfers}: $\sum_{i = 1}^n Y_i = Z$ for \emph{externally} sold position $Z \in L^\infty$.
\end{enumerate}
Note that if $Z \equiv 0$ then this modified equilibrium coincides exactly with the typical B\"uhlmann equilibrium.

Assume $Z \in L^\infty$ so that the modified equilibrium $(Y,\Q)$ exists (see Section 3.1 below for some discussion on this question). \blue{Implicitly, as an equilibrium solution, $\X + Z - \E^\Q[Z] \in \cl\D$.  For simplicity of exposition, we will assume that $\X + Z - \E^\Q[Z]\in\D$ a.s.\footnote{This condition can be relaxed to $\X + Z - \E^\Q[Z] \in \cl\D$ in case when $\D=\Rplus$, but requires a lengthy technical analysis with multiple clauses.}}

As presented in~\cite{buhlmann1984general} and also detailed in Appendix~\ref{sec:buhlmann-appendix}, the equilibrium probability measure $\Q$ must satisfy the fixed point problem:
\begin{equation}\label{eq:Q}
\frac{d\Q}{d\P}(\omega) = \frac{\exp\left(-\frac{1}{n}\int_{\essinf\X}^{\X(\omega) + Z(\omega) - \E^\Q[Z]} \rho(\gamma)d\gamma\right)}{\EP\left[\exp\left(-\frac{1}{n}\int_{\essinf\X}^{\X + Z - \E^\Q[Z]} \rho(\gamma)d\gamma\right)\right]}.
\end{equation}
Within the construction of the pricing measure $\Q$ from~\eqref{eq:Q} we, implicitly, consider $\rho\blue{>0}$ to be the harmonic average of risk aversions $\rho_i\blue{>0},~1\le i\le n$, i.e.,
\begin{equation}\label{eq:rho}
\rho(\gamma) = n\left(\sum_{i = 1}^n -\frac{u_i'(\Y_i(\gamma))}{u_i''(\Y_i(\gamma))}\right)^{-1} = n\left(\sum_{i = 1}^n \frac{1}{\rho_i(\Y_i(\gamma))}\right)^{-1},
\end{equation}
where $\Y_i,~1\le i\le \blue{n}$ solve a differential system with equilibrium initial conditions.  
\blue{Implicitly, as an equilibrium solution, we have that $\E\left[\exp\left(-\frac{1}{n}\int_{\essinf\X}^{\X + Z - \E^\Q[Z]} \rho(\gamma)d\gamma\right)\right] > 0$ {(noting that, in the case of $\D = \R_{++}$, $\int_{\essinf\X}^0 \rho(\gamma)d\gamma := \lim\limits_{z\searrow 0} \int_{\essinf\X}^z \rho(\gamma)d\gamma \in \R \cup \{-\infty\}$ by $\rho > 0$)}.}
We refer to Appendix~\ref{sec:buhlmann-appendix} and~\cite{buhlmann1984general,aase1993equilibrium,iwaki2001economic} for details of these constructs as well as the individual risk transfers $Y_i$. We will refer to $\rho$ as the risk aversion of the {\it harmonic representative agent}.

\begin{proposition}\label{prop:representative}
Assume that the (utility) functions $u_i: \bbd \to \R,~1\le i\le n$ are
twice differentiable, strictly increasing and concave.
If $\rho_i(z)=  -u_i''(z)/u_i'(z),~z\in\bbd$, is nonincreasing for every $1\le i \le n$, then $\rho: \bbd \to \R_+$ given in~\eqref{eq:rho} is nonincreasing.
\end{proposition}


As will be investigated in greater detail below, we are interested in pricing these liquidated contingent claims $Z \in L^\infty$ through the B\"uhlmann market mechanism.  By construction, the value of this contingent claim is given by $\E^\Q[Z]$ where $\Q$ is the B\"uhlmann pricing measure.  Utilizing~\eqref{eq:Q}, this price can be seen to satisfy the fixed point condition
\begin{equation}\label{eq:buhlmann-price}
\E^\Q[Z] = \frac{\EP\left[Z \exp\left(-\frac{1}{n}\int_{\essinf\X}^{\X+Z-\E^\Q[Z]} \rho(\gamma)d\gamma\right)\right]}{\EP\left[\exp\left(-\frac{1}{n}\int_{\essinf\X}^{\X+Z-\E^\Q[Z]} \rho(\gamma)d\gamma\right)\right]}.
\end{equation}
It is this fixed point problem, and variations of it, \blue{that} are the primary focus of this work.

\section{Equilibrium market impacts}\label{sec:main}

Motivated by the B\"uhlmann equilibrium setup, let $R: \bbd \to \R$ be 
the integral of the absolute risk aversion of the harmonic representative agent (up to a multiplicative constant), i.e., 
\[R(z) = \frac{1}{n}\int_{\essinf\X}^z \rho(\gamma)d\gamma.\]
\blue{Similar to the extensions taken above, if $\D=\Rplus,$ we extend $R(0) = \lim\limits_{z\searrow0} R(z)$ in the broader sense of the limit.}
In this way, the value $\E^\Q[Z]$ of the claim $Z \in L^\infty$ originally provided in the fixed point condition~\eqref{eq:buhlmann-price} can be written as:
\begin{equation}\label{eq:H}
\E^\Q[Z] = \frac{\EP\left[Z \exp\left(-R(\X+Z-\E^\Q[Z])\right)\right]}{\EP\left[\exp\left(-R(\X+Z-\E^\Q[Z])\right)\right]}.
\end{equation}
In fact, the function $R$ is bijective with the set of strictly increasing and concave utility functions $u$ for the harmonic representative agent (with equivalence class defined up to multiplicative and additive constants): 
\[u(z) = \int_0^z \exp(-nR(\essinf\X+y))dy.\]
We also note that, using the fact that harmonic mean is bounded from below by its minimum, $R$ is concave if the harmonic representative agent has a nonincreasing absolute risk aversion (in wealth).
In fact, Proposition~\ref{prop:representative} provides conditions on the set of $n$ market participants that guarantees the concavity of $R$.  The basic properties of $R$ are encoded within the following assumption we impose for the remainder of this work.
\begin{assumption}\label{assump:R}
For the rest of this paper we will assume that $R: \bbd \to \R$ is strictly increasing, differentiable, and concave. 
\end{assumption}

For the remainder of this work we will focus on and utilize this generalized function $R$ to encode the financial market and the harmonic representative agent's utility $u$.

The specific study of this work, rather than the modified B\"uhlmann equilibrium itself, is to determine the price and value of a contingent claim $Z \in L^\infty$:
\begin{align}
\label{eq:V} 
V(Z) &= \FIX_v \left\{H_Z(v) := \frac{\EP\left[Z \exp\left(-R(\X + Z - v)\right)\right]}{\EP\left[\exp\left(-R(\X + Z - v)\right)\right]}\right\} \blue{:= \left\{v \in \blue{\dom H_Z}\; | \; v = H_Z(v)\right\}}.
\end{align}
\blue{Here we set the domain of $H_Z$ to be 
\begin{equation*}
\dom H_Z = \begin{cases} \R &\text{if } \D = \R \\ \Big\{v \in \mathbb{R} \; \Big| \; \P(\X+Z \ge v) = 1, \; \E[\exp(-R(\X+Z-v))] > 0\Big\} &\text{if } \D = \Rplus \end{cases} 
\end{equation*}
in order to be consistent with the above formulation. }

This pricing function $V$ can be seen as satisfying the fixed point condition of~\eqref{eq:buhlmann-price}. We will also study two special cases ($R(x) = \alpha (x-\essinf\X)$ and $R(x) = \eta \log(x/\essinf\X)$ \blue{corresponding to the exponential and power utility settings respectively}) in Section~\ref{sec:cs} below which correspond directly with the modified B\"uhlmann setup under specific choices of utility functions.

In addition to the pricing function $V$, we also wish to consider the inverse demand functions generated by this market.  That is, given a portfolio $q$ being liquidated in the market, we wish to find the marginal price $f^q(s)$ for the $s^{th}$ unit sold and the average price for those same units $\bar f^q(s)$.  Such inverse demand functions satisfy the relation:
\begin{equation}\label{eq:Vf}
V(sq) = \int_0^s f^q(\gamma)d\gamma = s \bar f^q(s).
\end{equation}
It is these inverse demand functions that are often introduced and presented in the literature with $V$ derived through the relations of~\eqref{eq:Vf}.  For instance, we refer to~\cite{AFM16} as an important work on fire sales in systemic risk which provides sufficient results on the uniqueness of (external) system liquidations through the application of monotonicity conditions on the volume weighted average price $\bar f^q$ and $s \mapsto s\bar f^q(s)$.  However, in this equilibrium setup of market impacts, we find the construction of the pricing function $V$ from~\eqref{eq:V} to be more natural; in Section~\ref{sec:idf}, we study the inverse demand functions $f^q$ and $\bar f^q$ derived from $V$.


\subsection{Pricing function}\label{sec:value}
Consider a generalized structure from the B\"uhlmann setup in~\eqref{eq:V}. We first investigate the existence of the unique solution to the fixed point problem \eqref{eq:V}.  That is, we study the conditions so that the market is capable of providing a well-defined price for a contingent claim.  This is presented in Theorem~\ref{thm:unique} and expanded in Corollary~\ref{cor:Rplus}.  The equilibrium pricing problem~\eqref{eq:V} endogenizes the market impacts due to the limited liquidity and preferences of the market participants; this is in contrast to the exogenous valuation taken in\blue{, e.g.,~\cite{CFS05,GLT15,AFM16,braouezec2017strategic}} (via assumed forms of the inverse demand function). 

In order to approach the problem of endogenous pricing, we first present a general result on the existence of a fixed point of $H_Z$.
\begin{theorem}\label{thm:exists}
Assume $R$ satisfies Assumption \ref{assump:R} and let $Z \in L^\infty$.
Then there exists a solution \blue{to \eqref{eq:V}, i.e.\  $V(Z) \ne \emptyset$}
\blue{if one of the following conditions hold}:
\begin{enumerate}
\item $\D = \R$;

\item $\D = \Rplus$ and \blue{if there exists $v_0 \in \dom H_Z$ such that $H_Z(v_0) \le v_0$. This latter property holds if either} 
\begin{itemize}
\item 
$H_Z(\essinf[\X+Z]) \leq \essinf[\X+Z]$, when \blue{$\essinf[\X+Z]\in \dom H_Z$}, 
\\ \blue{or} 
\item $\liminf_{v \nearrow \essinf [\X + Z]} H_Z(v) < \essinf[\X + Z]$ \blue{otherwise (i.e., if $\essinf[\X+Z]\not\in \dom H_Z$)}. 
\end{itemize}
\end{enumerate}
\blue{Moreover, the set $V(Z)$ is compact.} 
\end{theorem}

As Theorem~\ref{thm:exists} provides simple conditions for the existence of a fixed point, we now wish to consider the question of uniqueness.  Notably, and not surprisingly, uniqueness of the fixed point requires stronger conditions than existence.  We refer the interested reader to Example~\ref{ex:nonunique} presented in the Online Appendix which provides a simple example in which there exist\blue{s} a multiplicity of equilibria.  Theorem~\ref{thm:unique}, however, provides sufficient conditions for the uniqueness of an equilibrium valuation.
\begin{theorem}\label{thm:unique}
Assume $R$ satisfies Assumption \ref{assump:R} and let $Z \in L^\infty$. \blue{Let either $\bbd = \R$ or $\bbd = \Rplus$.} There exists at most one solution \blue{to \eqref{eq:V}, i.e.\  $\abs{V(Z)}\le 1$} 
if any of the following conditions is satisfied:
\begin{enumerate}[(a)]
\item\label{thm:unique-comonotone} $Z$ and $\X+Z$ are comonotonic;
\item\label{thm:unique-linear} $R$ is linear;
\item\label{thm:unique-monotone} $z\in\Rplus \mapsto z \exp(-R(\X + z))$ is almost surely non-decreasing;
\item\label{thm:unique-concave} $z\in\Rplus \mapsto z \exp(-R(\X + z))$ is almost surely concave.
\end{enumerate}
\end{theorem}

\blue{
\begin{remark}\label{rem:deterministic}
Theorem~\ref{thm:unique}\eqref{thm:unique-comonotone} implies that there exists at most one solution  \blue{to \eqref{eq:V} for any $Z \in L^\infty$, i.e.\  $\abs{V(Z)}\le 1$} 
 if the aggregate holding $\X \in \bbd$ \blue{is} \emph{deterministic}.
\end{remark}
}

\begin{remark}\label{rem:R2}
Recall that $R$ \blue{is} the integral of the absolute risk aversion of the harmonic representative agent \blue{up to a multiplicative constant}. Condition~\eqref{thm:unique-monotone} of Theorem~\ref{thm:unique} can, thus, be viewed with respect to the risk aversion of the agent if $\essinf\X > 0$.  That is, $z\in\Rplus \mapsto z\exp(-R(\X+z))$ is a.s.\ nondecreasing if and only if $zR'(\X+z) \leq 1$ a.s.  In particular if $\essinf\X > 0$, by concavity \blue{of $R$}, this is true if $zR'(z) \leq 1$ for every $z > 0$.  Therefore, uniqueness of $V$ is guaranteed if the relative risk aversion of the harmonic representative agent is bounded from above by $\blue{n}$.
\end{remark}

In fact, if any of the conditions \eqref{thm:unique-comonotone}-\eqref{thm:unique-concave} of the above theorem is satisfied, we can strengthen the existence result condition of Theorem \ref{thm:exists}. Namely, in case $\D = \Rplus$, the sufficient condition for existence presented in Theorem~\ref{thm:exists} is also necessary for existence. 
\begin{corollary}\label{cor:exist-cond}
Assume $R$ satisfies Assumption \ref{assump:R} with $\D = \Rplus$ and let $Z \in L^\infty$.  If any of the conditions of Theorem~\ref{thm:unique} is satisfied then there exists a unique solution \blue{to \eqref{eq:V} , i.e.\  $\abs{V(Z)}= 1$}
if and only if $H_Z(\essinf[\X+Z]) \leq \essinf[\X+Z]$, when \blue{$\essinf[\X+Z]\in\dom H_Z$,} or $\liminf\limits_{v \nearrow \essinf[\X+Z]} H_Z(v) < \essinf[\X+Z]$ if otherwise  \blue{$\essinf[\X+Z]\not\in\dom H_Z$}. Furthermore, this solution is bounded from above by $\essinf[\X+Z]$, i.e.\ $V(Z) \le \essinf[\X+Z].$
\end{corollary}

Rather than checking the value of $H_Z(\essinf[\X+Z])$ for existence of a fixed point for \eqref{eq:V} (and therefore the pricing measure $\Q$) if $\bbd = \Rplus$, we wish to provide an alternative sufficient condition for existence.
\begin{corollary}\label{cor:Rplus}
Assume $R$ satisfies Assumption \ref{assump:R} with $\bbd = \Rplus$ and let $Z \in L^\infty$.
If 
\begin{align}
\EP\left[\exp\left(-R(\X+Z - \essinf [\X+Z])\right)\right]=\infty
\label{eq:infty-cond}
\end{align}
 then $Z\in\dom V$. 
\end{corollary}
This condition can be used to guarantee existence of a clearing price in many examples with $\bbd = \Rplus$.  Indeed, if $\P(\X + Z = \essinf[\X+Z]) > 0$ and $\lim_{z \searrow 0} R(z) = -\infty$ then $Z \in \dom V$.  As will be considered in Corollary~\ref{cor:bernoulli} below, the case in which $\X + Z$ attains its essential infimum with positive probability can be thought of as an extreme systemic shock scenario.  Specifically, under a systemic shock, all assets would attain their worst case scenario, i.e., their essential infimum.  Corollary~\ref{cor:Rplus} can be used to demonstrate that there exists a clearing price if a systemic shock is possible.

\subsection{Extension and Selection of $V$}\label{sec:extension}
We have now successfully defined $V$ on $L^\infty$ in case $\D=\R$, but have additional requirements if $\D=\Rplus$ to ensure existence -- specifically, we must at least have that either $\liminf_{v \nearrow \essinf [\X + Z]} H_Z(v) < \essinf [\X + Z]$ or, provided \blue{$\essinf[\X+Z] \in \dom H_Z$}, $H_Z(\essinf [\X + Z]) \leq \essinf [\X + Z]$ in order to define $V$. 
Additionally, except under the conditions of Theorem~\ref{thm:unique}, we cannot guarantee that the price of a claim $Z \in \dom V$ is uniquely defined.
Therefore, in this section we look for a way to extend the definition of $V$ to the entire space $L^{\infty}$ even in the case when $\D=\Rplus$ and to select the appropriate price of $Z \in \dom V$ when there exists a multiplicity of prices. \blue{Mathematically, this extension is done for convenience so that the domain in all cases (whether $\D=\R$ or $\D=\Rplus$) can be the entire space $L^{\infty}$.}

Conceptually, this extension and selection problem are fundamentally distinct cases.  The former occurs for claims $Z \in L^\infty$ such that $V(Z) = \emptyset$; the latter occurs for claims $Z \in L^\infty$ such that $V(Z)$ has cardinality at least 2.  The general notion for determining the \emph{unique} offered price for $Z$ in any case follows from two notions:
\begin{itemize}
\item market participants compete for underpriced claims which can be used to improve utility and
\item market participants prefer paying less to more.
\end{itemize}
Though the case of $Z \not\in\dom V$ may not yield a ``fair'' price in the B\"uhlmann sense, we assume, as in the fire sale literature~\cite{CFS05,AFM16,feinstein2015illiquid}, that the external seller of $Z$ is forced to complete the liquidation and thus must accept the value provided by the market.  Similarly, if there exist multiple equilibrium prices for the claim, the external seller must accept whichever fair price is provided by the market.

Mathematically, this extension and selection of the pricing function $V$ is given by $\bar V: L^\infty \to \R$ such that, for payoff $Z\in L^\infty$, 
\begin{equation}
\label{eq:barV} \bar V(Z) := \begin{cases} \min V(Z) &\text{if } Z \in \dom V, \\ \essinf[\X+Z] &\text{if } Z \not\in \dom V. \end{cases}
\end{equation}
\blue{We now argue that the way this extension and selection must be set as in \eqref{eq:barV} to be financially meaningful. This is extended further within Lemma~\ref{lemma:barV} below to demonstrate that $\bar V$ satisfies some desired mathematical and financial properties.}
This definition encodes exactly the conceptual notions provided above.  If multiple clearing prices are available, the minimal such price is selected as that is preferred by all market participants;\footnote{The minimum $\min V(Z)$ is well defined since $V(Z)$ is compact for any $Z \in L^\infty$ as follows from Theorem \ref{thm:exists}.} this is similar to an English auction, with the price starting at $\essinf Z$, and rising to $\min V(Z)$ which is the smallest equilibrium, and where it therefore remains. Note that this is the lowest price that no market participant finds the claim underpriced.  On the other hand, if no ``fair'' price exists (which can only occur if $\D = \Rplus$ by Theorem~\ref{thm:exists}) then the maximum price the market participants are both willing and are able \blue{to} pay is offered; more specifically, if $Z\not\in\dom V$ then the market finds that it does not have sufficient liquidity to pay a fair (clearing) value for the claim $Z$ as $H_Z(v) > v$ for every $v \in [\essinf Z,\essinf[\X+Z])$\blue{, as follows from Theorem \ref{thm:exists}}.  Since the market is only limited by the \blue{risk-free capital} available to it, this means the market will, instead, offer $\essinf[\X+Z]$ for the claim $Z$ (which is a good deal for the market participants and will be accepted due to the assumption that the external liquidation of the claim is forced on the market).
Finally, and trivially, if $Z\in\dom V$ such that $V(Z)$ is a singleton (e.g., under the conditions of Theorem~\ref{thm:unique}), then $\bar V(Z)$ provides exactly this price.

With this definition of the extension and selection $\bar V$ of $V$, we now wish to consider some intuitive properties of $\bar V$.  Namely, we prove that $\bar V$ is bounded, law invariant, cash translative, and (lower semi)continuous. We also formulate additional properties for $\bar V$ under certain conditions on $R$.  Specifically, we provide conditions for the monotonicity and concavity of $\bar V$; that is, respectively, greater liquidations provide a larger value and the marginal increase in value is decreasing.  We note that, taken together, these properties of $\bar V$ construct a monetary risk measure (with modification up to negative signs) which may be of interest for future study.
\begin{lemma}\label{lemma:barV}
Assume $R$ satisfies Assumption~\ref{assump:R}.
\begin{enumerate}
\item\label{lemma:barV-bound} $\bar V(Z) \in [\essinf Z, \esssup Z]$ for any $Z \in L^\infty$. 
\item\label{lemma:barV-law} $\bar V(Z_1) = \bar V(Z_2)$ if $(Z_1,\X+Z_1) \overset{(d)}{=} (Z_2,\X+Z_2)$, i.e.\ equality in distribution.
\item\label{lemma:barV-translative} $\bar V(Z+z) = \bar V(Z)+z$ for any $Z \in L^\infty$ and $z \in \R$. 
\item\label{lemma:barV-cont} $\bar V$ is lower semicontinuous in the strong topology.\footnote{The strong topology is the normed topology with norm $\|\cdot\|_{\infty}$, i.e., $(Y_m)_{m \in \mathbb{N}} \to Y \in L^\infty$, if $\lim_{m \to \infty} \|Y_m - Y\|_{\infty} = 0$.}  If any of the conditions of Theorem~\ref{thm:unique} hold then $\bar V$ is continuous in the strong topology.
\item\label{lemma:barV-monotone} \blue{If $z \in \Rplus \mapsto z\exp(-R(\X+z))$ is a.s.\ nondecreasing then $\bar V(Z_1) \geq \bar V(Z_2)$ for $Z_1 \geq Z_2$ a.s.\ and $\bar V$ is Lipschitz continuous with Lipschitz constant $1$ with respect to the maximum norm.}
\item\label{lemma:barV-concave} \blue{If $z \in \Rplus \mapsto z \exp(-R(\X+z))$ is a.s.\ concave then $\bar V$ is concave and upper semicontinuous in the weak* topology.}\footnote{The weak* topology is the coarsest topology such that $Y \in L^\infty \mapsto \EP[Y^* Y]$ is continuous for any integrable random variable $Y^* \in L^1$, i.e., $(Y_j)_{j \in J} \to Y \in L^\infty$ (for index set $J$) if $\lim_{j \in J}\EP[Y^* Y_j] = \EP[Y^* Y]$ for every $Y^* \in L^1$.}
\end{enumerate}
\end{lemma}
\begin{remark}
Note that the extra conditions of Lemma~\ref{lemma:barV}\eqref{lemma:barV-monotone} and~\eqref{lemma:barV-concave} guarantee that there exists at most a single fair price $V(Z)$ for any claim $Z \in L^\infty$ as proven in Theorem~\ref{thm:unique}.
\end{remark}
It turns out the same extension $\bar V$ can also be achieved using a wholly different notion provided $\bbd = \Rplus$ and $\lim\limits_{z \to 0} R(z) = -\infty$ under our conditions for uniqueness from Theorem~\ref{thm:unique}. 
\blue{For the rest of this subsection,} we wish to highlight the dependency of $V$ on $\X$, and therefore write $V(Z;\X)$. 
Let $B[p]\sim Bern(p)$ be a Bernoulli random variable representing the ruin (with probability $1-p$) of the banking system. \blue{It is natural to assume that in case of ruin the assets will pay their minimum, and in particular the claim will pay $\essinf Z$ and the aggregate endowment will be worth $\essinf \X$.}
It follows that the position $Z\in L^{\infty}$ becomes $B[p](Z-\essinf Z)+\essinf Z$ as it is only payable if the system has not defaulted; similarly the assets of all market participants would be subject to the same systemic stress, i.e., $\X$ becomes $B[p](\X-\essinf\X)+\essinf\X$.  As $p$ tends towards 1, i.e., the probability of systemic ruin tends towards 0, the claim $Z$ and aggregate assets $\X$ are again recovered. 
\blue{Indeed, recall our assumption that $R(0)\define \lim\limits_{z \to 0} R(z) = -\infty.$ 
Therefore the condition~\eqref{eq:infty-cond} of Corollary~\ref{cor:Rplus} holds for any $p\in(0,1)$.  That is, under these systemic shocks $B[p](Z-\essinf Z)+\essinf Z \in \dom V(\cdot ; B[p](\X-\essinf\X)+\essinf\X)$ for any $p \in (0,1)$.
}
Armed with this observation, we formulate an alternative definition for $\bar V$ under the uniqueness conditions of Theorem~\ref{thm:unique}:
\begin{equation}
\label{eq:barV-bernoulli} \hat V(Z;\X) = \lim_{p \nearrow 1} V(B[p](Z-\essinf Z) + \essinf Z; B[p](\X-\essinf\X) +\essinf\X).
\end{equation}
That is, up to modification via the essential infimum, the (extended) price $\hat V(Z;\X)$ of $Z$ is provided by the limiting behavior of the price of $B[p]Z$ under market assets $B[p]\X$ as the probability of systemic ruin tends towards 0.  \blue{The limit within~\eqref{eq:barV-bernoulli} is guaranteed to exist by monotone convergence; this result is formalized within the proof of Corollary~\ref{cor:bernoulli}.}
The following result shows that $\hat V = \bar V$.  Thus, this new extension $\hat V$ provides the interpretation that the market prices $Z$ as if the probability of systemic ruin is negligibly small rather than the explicitly setting the probability to 0. 

\begin{corollary}\label{cor:bernoulli}
Assume $R$ satisfies Assumption \ref{assump:R}. Assume also that $\D=\Rplus$ and $\lim_{z \to 0} R(z) = -\infty$. 
Let $Z \in L^\infty$ and assume that $B[p]$ in the construction \eqref{eq:barV-bernoulli} is independent of $\X,Z$ for every $p \in (0,1)$. 
If any of the conditions of Theorem~\ref{thm:unique} is satisfied and $\essinf[\X+Z] = \essinf\X + \essinf Z$, then $\hat V(Z;\X)= \bar V(Z;\X)$. 
\end{corollary}

\begin{remark}\label{rem:essinf}
Within Corollary~\ref{cor:bernoulli}, we introduced the additional condition that $\essinf[\X+Z] = \essinf\X + \essinf Z$.  This condition can be considered as the limiting condition for the systemic ruin encoded in the Bernoulli random variables as the probability of ruin tends to $0$.  That is, the worst case for the market ($\essinf\X$) and for the claim $Z$ ($\essinf Z$) coincide as both are stressed by the same systemic shock.
\end{remark}

\subsection{Inverse demand functions}\label{sec:idf}
Now that we have a good definition of a unique price $\bar V$ that is valid over the entire space $L^\infty$, we are finally able to rigorously define the inverse demand functions. There is no unique way to do so, and we choose to demonstrate how this can be done following the example set in \cite{banerjee2019price,bichuch2020repo}: using the \emph{order book density} and the \emph{volume weighted average price} (\emph{VWAP}) function. For the former, we set $f^q: \R_+ \to \R$, so that $\bar V(sq) = \int_0^s f^q(t)dt$, i.e., $f^q(s) = \frac{\d}{\d s}\bar V(sq)$. For the latter, we define $\bar f^q: \R_+ \to \R$ by 
\begin{align}
\bar f^q(s) = \begin{cases} \frac{\EP[q\exp(-R(\X))]}{\EP[\exp(-R(\X))]} &\text{if } s = 0, \\ \bar V(sq)/s &\text{if } s > 0. \end{cases}
\label{eq:VWAP}
\end{align}
The order book density function $f^q$ provides the price of the next marginal unit of the portfolio $q$ given the total number of units already sold; in this way we can encode a dynamic notion of pricing in a single period framework.  In contrast, the VWAP function $\bar f^q$ provides the average price per unit of liquidated portfolio; this construction is implicitly utilized in much of the fire sale literature, see, e.g.,~\cite{CFS05,AFM16,feinstein2015illiquid}.

As discussed previously, these inverse demand functions ($f^q$ or $\bar f^q$) are often the objects introduced exogenously. 
Such an approach, though valid, does not necessarily follow from a financial equilibrium. 
By first studying the value of arbitrary portfolios, we are able to consequently talk about the price per unit of any asset or portfolio.  
\blue{As presented in Remark \ref{rem:multi-asset}, we can consider the cross-impacts that liquidating one asset can have on the price of another; this is in contrast to the typical, simplifying, assumption that there are no cross-impacts as taken in, e.g.,~\cite{GLT15}. This is presented explicitly} in the special cases of
markets generated by the exponential or power utility functions (as detailed in Section~\ref{sec:cs} below). 

\begin{assumption}\label{ass:essinf}
Throughout this section, we assume $\X,q$ satisfy the joint systemic ruin condition $\essinf[\X+q] = \essinf\X + \essinf q$ (so that $X,sq$ satisfy this condition for every $s \geq 0$) as introduced in Corollary~\ref{cor:bernoulli} and discussed in Remark~\ref{rem:essinf}.
\end{assumption}

We first show that the order book density function $f^q$ is well defined.  That is, we can meaningfully discuss the price of the next \emph{marginal} unit of the portfolio $q$.
\begin{lemma}\label{lemma:f} 
Assume $R$ satisfies Assumption~\ref{assump:R} and any of the conditions of Theorem~\ref{thm:unique} hold.  Consider the setting in which a single portfolio is being liquidated proportionally, i.e., $f^q: \R_+ \to \R$ defined by
\[f^q(s) = \begin{cases} \frac{\EP[q(1-[sq - \bar V(sq)]R'(\X+sq-\bar V(sq)))\exp(-R(\X+sq-\bar V(sq)))]}{\EP[(1-[sq - \bar V(sq)]R'(\X+sq-\bar V(sq)))\exp(-R(\X+sq-\bar V(sq)))]} &\text{if } sq \in \dom V \\ \essinf q &\text{else}. \end{cases}\]
\begin{enumerate}
\item\label{lemma:f-exist} $\bar V(sq) = \int_0^s f^q(t)dt$ for any $s \in \R_+$.
\item\label{lemma:f-bound} If $z\in\Rplus \mapsto z\exp(-R(\X+z))$ is nondecreasing then $f^q(s) \geq \essinf q$ for every $s \in \R_+$ and if, additionally, $z \mapsto z\exp(-R(\X+z))$ is strictly increasing and $\P(q > \essinf q) > 0$ then $f^q(s) > \essinf q$ for every $s \in \R_+$ such that $sq \in \dom V$.
\item\label{lemma:f-cont} $f^q$ is continuous on $\operatorname{int}\{s \in \R_+ \; | \; sq \in \dom V\}$ if $R$ is continuously differentiable.
\item\label{lemma:f-monotone} $f^q$ is nonincreasing if $z\in\Rplus \mapsto z\exp(-R(\X+z))$ is nondecreasing and concave.
\end{enumerate}
\end{lemma}

Now we want to consider the volume weighted average price $\bar f^q$.  As shown below, this pricing function satisfies the expected conditions automatically in contrast to the order book density function $f^q$.  This VWAP function provides exactly the average price obtained by the seller per unit of the portfolio $q$, i.e., $\bar f^q(s) = \frac{1}{s}\bar V(sq)$ for $s > 0$. 
\begin{lemma}\label{lemma:barf} 
Assume $R$ satisfies Assumption~\ref{assump:R} and any of the conditions of Theorem~\ref{thm:unique}.  Consider the setting in which a single portfolio $q$ is being liquidated proportionally, i.e., $\bar f^q: \R_+ \to \R$ is given by \eqref{eq:VWAP}.
\begin{enumerate}
\item\label{lemma:barf-bound} $\bar f^q(s) \geq \essinf q$ for every $s \in \R_+$ and if, additionally, $\P(q > \essinf q) > 0$ then $\bar f^q(s) > \essinf q$ for every $s \in \R_+$.
\item\label{lemma:barf-cont} $\bar f^q$ is continuous.
\item\label{lemma:barf-monotone} $\bar f^q$ is nonincreasing.
\end{enumerate}
\end{lemma}

\begin{remark}\label{rem:idf}
The order book density function $f^q$ and VWAP function $\bar f^q$ are related through:
\[\bar V(sq) = \int_0^s f^q(t)dt = s\bar f^q(s).\]
\end{remark}

\begin{remark}\label{rem:liquidity} 
As defined in, e.g.,~\cite{feinstein2019leverage,banerjee2019price}, the liquidity of the inverse demand function near 0 is defined as the velocity that prices are impacted by a small, additional, liquidation.  In those works, the liquidity was used to calibrate risk weights for studying fire sales subject to capital adequacy requirements.
Notably, by construction, the inverse demand functions coincide with the expectation of $q$ under the measure with Radon-Nikodym derivative $\exp(-R(\X))/\EP[\exp(-R(\X))]$ when no assets are being sold, i.e., $f^q(0) = \bar f^q(0) = \EP[q\exp(-R(\X))]/\EP[\exp(-R(\X))]$.  The liquidity of the market near 0, i.e., the impact of a small liquidation, can be provided explicitly by:
\begin{align}
(f^q)'(0) =& -2\frac{\EP[q^2R'(\X)\exp(-R(\X))]}{\EP[\exp(-R(\X))]}  + 4 \frac{\EP[q\exp(-R(\X))]\EP[qR'(\X)\exp(-R(\X))] }{\EP[\exp(-R(\X))]^2} \nonumber\\
&\quad-2\frac{ \EP[q\exp(-R(\X)) ]^2\EP[R'(\X)\exp(-R(\X))]}{\EP[\exp(-R(\X))]^3} , \nonumber\\
(\bar f^q)'(0) =& \frac12(f^q)'(0).\label{eq:bar-f-deriv}
\end{align}
%
These notions of market liquidity \blue{can be simplified} significantly if $\X\in\D$ deterministic. Under such a setting the order book density and the VWAP functions provide respectively the initial liquidity values of 
\begin{align}
(f^q)'(0) =2(\bar f^q)'(0)= -2R'(\X)\operatorname{Var}(q).
\label{eq:X-determin-ex}
\end{align}
 That is, the market clearing price drops no matter how it is measured when the first (marginal) unit is sold (unless $q$ is deterministic) and is proportional to both the absolute risk aversion of the harmonic representative agent at the market wealth $\X$ and to the variance of portfolio $q$.
\end{remark}


\blue{
Before continuing to the analytical results in Section~\ref{sec:cs}, we wish to briefly discuss the multi-asset inverse demand functions.
\begin{remark}\label{rem:multi-asset}
Comparable to the definition of the order book density and VWAP functions in the single portfolio setting, these inverse demand functions are defined for a setting with $m$ assets (with vector payoffs $\q = (q_1,...,q_m)^\T$) such that $\bar V(\s^\T \q) = \int_{0}^1 \f^{\q}(\vec{r}(t))\vec{\dot r}^\T(t)dt = \s^\T \bar\f^{\q}(\s)$ for order book density function $\f^{\q}: \R^m_+ \to \R^m$ and VWAP $\bar\f^{\q}: \R^m_+ \to \R^m$ and such that $\vec{r}\colon[0,1]\to \R^m$ denotes an arbitrary curve from $\vec r(0)=\vec{0}$ to $\vec r(1)=\s$.  
This curve $\vec{r}$ describes the order of sales over ``time''. However, in our single-period setting, the timing of financial transactions does not matter; it is for this reason that the curve $\vec{r}$ is arbitrary.  Without loss of generality, we take the curve $\vec r(t) := \s t,~t \in [0,1]$.  Our aim herein is to match these definitions of the order book density function $\f^{\q}$ and VWAP $\bar\f^{\q}$ to the original definitions used throughout this section.  
Specifically, we 
define
$\f^{\q}(\s) := \nabla_s \bar V(\s^\T\q)$ to be the gradient of $\bar V$ with respect to $\s$.  In this way, $\bar V(\s^\T \q) = \s^\T \int_0^1 \f^{\q}(\s t)dt$; this representation provides a specific choice for the VWAP $\bar\f^{\q}$ that is consistent with the order book density function, i.e., $\bar\f^{\q}(\s) := \int_0^1 \f^{\q}(\s t)dt$.  On the domain of $V$, these inverse demand functions can be provided more explicitly as:
\begin{align*}
\f^{\q}(\s) &= \frac{\E[\q (1 - [\s^\T \q + \bar V(\s^\T \q)] R'(\X + \s^\T \q - \bar V(\s^\T \q))) \exp(-R(\X + \s^\T \q - \bar V(\s^\T \q)))]}{\E[(1 - [\s^\T \q + \bar V(\s^\T \q)] R'(\X + \s^\T \q - \bar V(\s^\T \q))) \exp(-R(\X + \s^\T \q - \bar V(\s^\T \q)))]}\\
\bar\f^{\q}(\s) &= \frac{\E[\q \exp(-R(\X + \s^\T \q - \bar V(\s^\T \q)))]}{\E[\exp(-R(\X + \s^\T \q - \bar V(\s^\T \q)))]}
\end{align*}
for $\s \in \R^m_+$ such that $\s^\T \q \in \dom V$.
\end{remark}
}

\section{Special cases}\label{sec:cs} 

We now specialize the generic framework of the previous section to consider specific examples of the pricing and inverse demand functions. 
We highlight that these functions are the result of the modified B\"uhlmann equilibrium setting presented in Section~\ref{sec:buhlmann} with specific choices for the utility functions of the market participants.
Specifically, we consider two settings in the modified B\"uhlmann framework: exponential utility maximizers and power (or logarithmic) utility maximizers.
With these utility maximizing settings, we illustrate how asset and portfolio prices can be obtained as a result of the equilibrium problem.  In this way, in Example~\ref{ex:exponential} we are able to recover two classical inverse demand functions -- the linear and exponential inverse demand functions -- \blue{that} are commonly used.  This is significant as it makes explicit a number of assumptions that contribute to these (and other) specific pricing functions. 
Furthermore, not every inverse demand function can be achieved as an equilibrium from a specific set of market participants (as encoded in their utility functions and aggregate holdings $\X$).  \blue{For example, in case of deterministic $\X$, a power inverse demand function $f(s) = 1-s^p,~0<p<1$ cannot be achieved, as otherwise, $\lim\limits_{s\searrow0} f'(s) =- \infty$, which contradicts \eqref{eq:X-determin-ex} and \eqref{eq:bar-f-deriv} in the  order book density and VWAP cases respectively.
}

As such, extra care needs to be taken when using a specific, exogenous, inverse demand function, as it may not be achievable with the set of market participants or, potentially, from the specific distribution of the returns.

\begin{assumption}
For \blue{mathematical} simplicity, we will consider \emph{deterministic} aggregate holdings $\X \in \D$ only throughout this section.  Note this setting satisfies the uniqueness condition of Theorem~\ref{thm:unique}\eqref{thm:unique-comonotone} for any $R$ satisfying Assumption~\ref{assump:R} and any $Z \in L^\infty$.
\end{assumption}

\subsection{Exponential utility}
Consider the B\"uhlmann equilibrium construction from Section~\ref{sec:buhlmann} in which every market participant has exponential utility function $u_i(x) := 1 - \exp(-\alpha_i x)$ with risk aversion $\alpha_i > 0$.  Let $\alpha := \left(\sum_{i = 1}^n \frac{1}{\alpha_i}\right)^{-1}$. Then $\alpha$, up to a multiplication by $n$, is the harmonic average of the risk aversions.  As the absolute risk aversion in this case is constant, it immediately follows that we can construct a payment system with the function:
\begin{align}
R(x) &= \alpha (x-\X) \mbox{ with }\bbd = \R, \label{eq:R-lin}\\
\label{eq:V-exponential} V\left(Z;{\alpha}\right) &= \frac{\EP[Z\exp(-{\alpha} Z)]}{\EP[\exp(-{\alpha} Z)]}.
\end{align}
As $\bbd = \R$, it immediately follows that $\bar V \equiv V$.  We further wish to note that $V$ is defined as the Esscher premium \cite{godin2019general}.  
\begin{proposition}\label{prop:V-exponential}
Let the pricing function $V: L^\infty \to \R$ be defined as in~\eqref{eq:V-exponential}.  As more market participants enter the system, the market becomes more liquid (liquidation value goes up), i.e., $V(Z;\alpha_1) \leq V(Z;\alpha_2)$ for $\alpha_1 \geq \alpha_2 > 0$.
\end{proposition}
For a vector of assets $\q = (q_1,...,q_m)^\T$, the order book density function and VWAP inverse demand function are consequently defined as
\begin{align}
\label{eq:odf-exponential}
& \f^\q(\s) = \frac{\EP[\q\exp(-\alpha \s^\T \q)]}{\EP[\exp(-\alpha \s^\T \q)]} +\\
& \qquad\qquad\alpha\frac{\EP[(\s^\T\q)\exp(-\alpha\s^\T\q)]\EP[\q\exp(-\alpha\s^\T\q)] - \EP[\exp(-\alpha\s^\T\q)]\EP[(\s^\T\q)\q\exp(-\alpha\s^\T\q)]}{\EP[\exp(-\alpha\s^\T\q)]^2},\\
\label{eq:vwap-exponential} 
&\bar\f^\q(\s) = \frac{\EP[\q\exp(-\alpha \s^\T \q)]}{\EP[\exp(-\alpha \s^\T \q)]}.
\end{align}
As highlighted in the below proposition, these inverse demand functions exhibit no cross-impacts on the prices as is often assumed (see, e.g.,~\cite{GLT15,CS17,cont2019monitoring}) so long as the components of $\q$ are pairwise independent.
\begin{proposition}\label{prop:idf-exponential}
Let $\q$ be a $m$-dimensional random vector.  Let the order book density function $\f^\q: \R^m_+ \to \R^m$ and VWAP inverse demand function $\bar\f^\q: \R^m_+ \to \R^m$ be defined as in~\eqref{eq:odf-exponential} and~\eqref{eq:vwap-exponential} respectively. If the components of $\q$ are pairwise independent then the inverse demand functions exhibit no price cross-impacts, i.e., $f_k^\q(\s) = f_k^{q_k}(s_k)$ and $\bar f_k^\q(\s) = \bar f_k^{q_k}(s_k)$ for every $\s \in \R^m_+$ and $k = 1,...,m$.
\end{proposition}
\begin{proof}
This follows directly by independence and properties of the exponential function.
\end{proof}

Though 
\blue{\eqref{eq:V-exponential}}
provides a clear, analytical, structure for the pricing function $V$ and the inverse demand functions, it \emph{only} satisfies those properties that hold generally for such functions \blue{(i.e., Lemma~\ref{lemma:barV}\eqref{lemma:barV-bound}-\eqref{lemma:barV-cont})}.  For instance, it is \emph{not} true that $V$ is monotonic or concave in general (see, e.g., the discussion for the inverse demand function under a Poisson distribution \blue{in Example~\ref{ex:exponential}}).  Similarly, the order book density function $f^q$ does \emph{not} provide any clear structure as it can be negative (see, e.g., the discussion for the inverse demand function under either the normal or Poisson distributions).  However, those properties proven above \blue{within Lemmas~\ref{lemma:f} and~\ref{lemma:barf}} that hold generally (e.g., continuity and monotonicity of the VWAP inverse demand function $\bar f^q$) will hold herein.

The exponential utility setting provides an added benefit; so long as $\q$ has a distribution with a moment generating function, we can easily define the inverse demand functions $\f^\q$ and $\bar\f^\q$.  Notably, this allows us to consider a larger domain than $(L^\infty)^m$, including, e.g., the multivariate normal distribution.  
In those cases, motivated by the previous sections, we will define $V$ using \eqref{eq:V}, which in this setting (with exponential utility and linear $R$ in \eqref{eq:R-lin}) simplifies to \eqref{eq:V-exponential}. We then compute $\f^\q$ and $\bar\f^\q$ \blue{by utilizing \eqref{eq:V-exponential} to compute $V$ and using known closed form moment generating functions.} 

In Example~\ref{ex:exponential} below, we will consider a few well known distributions to provide the structure of the inverse demand functions.  We also wish to highlight a final, discrete, distribution that provides a counterexample for the convexity of these inverse demand functions in general. 
\begin{example}\label{ex:exponential}
Let $\f^\q$ and $\bar\f^\q$ be constructed from the exponential utility function.  Under the following distributions of $\q$ we can determine the structure of the inverse demand functions explicitly.
\begin{itemize}
\item \textbf{Multivariate normal}: If $\q \sim N(\mu,C)$ then $\f^\q(\s) = \mu - 2\alpha C\s$ and $\bar\f^\q(\s) = \mu-\alpha C \s$.  That is, we recover the linear inverse demand function common in the literature (see, e.g., \cite{GLT15,braouezec2017strategic,cont2019monitoring}). This makes this example one of the most important examples presented here. 
Notably, both the order book density and the VWAP inverse demand functions in this setting can take negative prices.  This is unsurprising given the positive probability of negative payoffs for these assets.  Financially, the possible negativity of the order book density function implies there exists some market depth beyond which the appetite for risk is satiated and further liquidations require compensating the purchasing counterparty for accepting this risk.
\item \textbf{Poisson}: If $q \sim \operatorname{Pois}(\lambda)$ then $f^q(s) = (1-\alpha s)\lambda \exp(-\alpha s)$ and $\bar f^q(s) = \lambda\exp(-\alpha s)$.  That is, we recover the exponential inverse demand function common in the literature (see, e.g., \cite{AW_15,feinstein2019leverage}). This is the second most important example in the paper. We, not only, recovered the exponential inverse demand function, but this example also shows that the order book density can become negative even for nonnegative payoffs, and that it need not be convex for all $s\ge0$. As a consequence of the negativity of $f^q$ \blue{(i.e., the marginal price dropping below the essential infimum of the claim's payoff which introduces the possibility for risk-free profits)}, we find that the order book has a finite depth $\alpha^{-1}$ such that any greater liquidations ($s > \alpha^{-1}$) no longer has a meaningful financial interpretation in terms of the order book.
\item \textbf{Bernoulli}: If $q \sim \operatorname{Bern}(p)$ then $f^q(s) = \frac{p^2 + (1-\alpha s)p(1-p)\exp(\alpha s)}{(p + (1-p)\exp(\alpha s))^2}$ and $\bar f^q(s) = \frac{p}{p + (1-p)\exp(\alpha s)}$.  
\item \textbf{Gamma}: If $q \sim \Gamma(k,\theta)$ then $f^q(s) = \frac{k\theta}{(1+\alpha\theta s)^2}$ and $\bar f^q(s) = \frac{k\theta}{1+\alpha\theta s}$. 
\item \textbf{Discrete distribution}: We wish to conclude with a simple distribution that results in nonconvex inverse demand functions $f^q$ and $\bar f^q$.  Let $q \in \{0,1,16\}$ with $\P(q = 0) = 0.02, \P(q = 1) = 0.49, \P(q = 16) = 0.49$. Then $\operatorname{skew}(q) = \frac{16.5584}{56.5411^{3/2}} = 0.0389 > 0$ which implies both $f^q$ and $\bar f^q$ are convex near $0$.  Consider now $s = \frac{1}{\alpha}$, $\bar f^q(s) \approx 0.90$ but $(\bar f^q)''(s) = \alpha^2 \E^\Q[(q-\bar f^q(s))^3] \approx -0.071\alpha^2 < 0$. Similarly it can be shown that the order book density function $f^q$ is nonconvex as well. 
\end{itemize}
\end{example}

\subsection{Power utility}\label{sec:cs-power}
Consider the B\"uhlmann equilibrium construction from Section~\ref{sec:buhlmann} in which every market participant has power utility function $u_i(x) := \frac{x^{1-\eta}-1}{1-\eta}$ if $\eta \neq 1$ and $u_i(x) = \log(x)$ if $\eta = 1$ for constant relative risk aversion $\eta \geq 0$.  As the relative risk aversion in this case is constant, it immediately follows that we can construct a payment function with the function $R(x) = \eta(\log(x)-\log(\X))$ (with $\bbd = \Rplus$), i.e.,
\begin{align}
\label{eq:V-power} V(Z) &= \FIX_v \frac{\EP[Z(\X + Z - v)^{-\eta}]}{\EP[(\X + Z - v)^{-\eta}]}.
\end{align}
We wish to note that, generally, there is no closed form for this construction.  As discussed previously, due to $\bbd = \Rplus$, the domain of $V$ is not the entire space $L^\infty$.  In fact, as proven in Proposition~\ref{prop:V-power} below, $\dom V = \{Z \in L^\infty \; | \; \EP[(Z-\essinf Z)^{1-\eta}] \leq \X\EP[(Z-\essinf Z)^{-\eta}]\}$.  In addition to the properties listed below, the extension $\bar V$ also satisfies all properties that hold generally (e.g., translativity).
\begin{proposition}\label{prop:V-power}
Let the pricing function $V$ be defined as in~\eqref{eq:V-power} for some $\eta \in [0,1]$ and let $\bar V$ be its extension as defined in~\eqref{eq:barV}.
\begin{enumerate}
\item $\dom V = \{Z \in L^\infty \; | \; \EP[(Z-\essinf Z)^{1-\eta}] \leq \X\EP[(Z-\essinf Z)^{-\eta}]\}$.
\item $\bar V$ is nondecreasing and concave.  It is, additionally, Lipschitz continuous in the strong topology and weak* upper semicontinuous.
\item If additionally $\eta > 0$ then $\bar V$ is strictly increasing and strictly concave on $\dom V$.
\end{enumerate}
\end{proposition}

With the construction for $\bar V$, we can consider the inverse demand functions; for simplicity we will first present the setting with a single portfolio liquidated proportionally.  However, these functions do not provide any analytical expression except one w.r.t.\ $\bar V$, i.e.,
\begin{align}
\label{eq:odf-power} f^q(s) &= \begin{cases} \frac{\EP[q(1-\eta\frac{sq-\bar V(sq)}{\X + sq-\bar V(sq)})(\X+sq-\bar V(sq))^{-\eta}]}{\EP[(1-\eta\frac{sq-\bar V(sq)}{\X + sq-\bar V(sq)})(\X+sq-\bar V(sq))^{-\eta}]} &\text{if } s \leq \frac{\X\EP[(q-\essinf q)^{-\eta}]}{\EP[(q-\essinf q)^{1-\eta}]} \\ 
     \essinf q &\text{else}, \end{cases} \\
\label{eq:vwap-power} \bar f^q(s) &= \begin{cases} \frac{\EP[q(\X + sq - \bar V(sq))^{-\eta}]}{\EP[(\X + sq - \bar V(sq))^{-\eta}]} &\text{if } s \leq \frac{\X\EP[(q-\essinf q)^{-\eta}]}{\EP[(q-\essinf q)^{1-\eta}]} \\ 
    \frac{\X}{s} + \essinf q &\text{else}. \end{cases}
\end{align}
Notably, due to the construction of $R$, the properties \blue{provided in Lemma~\ref{lemma:f}} for the order book density function $f^q$ (for a single portfolio being liquidated proportionally) hold for  $\eta\in[0,1]$.  The general properties of \blue{the VWAP} inverse demand function hold for every $\eta \geq 0$ \blue{by Lemma \ref{lemma:barf}}.
\begin{corollary}\label{cor:idf-power}
Consider the order book density function $f^q$ with a single portfolio being liquidated proportionally.  Assume $\P(q > \essinf q) > 0$.  If $\eta \in [0,1]$ then $f^q$ is nonincreasing and bounded from below by $\essinf q$.  If, additionally, $\eta > 0$ then $f^q(s) > \essinf q$ for every $s \in [0,\frac{\X\EP[(q-\essinf q)^{-\eta}]}{\EP[(q-\essinf q)^{1-\eta}]})$.
\end{corollary}
\begin{proof}
This is a direct consequence of Lemma~\ref{lemma:f}.
\end{proof}
We wish to highlight that, as opposed to the exponential utility case above, even if the assets are independent it is not guaranteed that $f_k^\q(\s)$ can be separated into a function $f_k^{q_k}(s_k)$.  This is clear from the construction of the inverse demand functions for a $m$-dimensional random vector $\q$, i.e., for $\s \in \R^m_+$ such that $\frac{\EP[(\s^\T\q - \essinf \s^\T\q)^{1-\eta}]}{\EP[(\s^\T\q - \essinf \s^\T\q)^{-\eta}]} \leq \X$:
\begin{align*}
\f^\q(\s) &= \frac{\EP[\q(1-\eta\frac{\s^\T\q-\bar V(\s^\T\q)}{\X + \s^\T\q-\bar V(\s^\T\q)})(\X+\s^\T\q-\bar V(\s^\T\q))^{-\eta}]}{\EP[(1-\eta\frac{\s^\T\q-\bar V(\s^\T\q)}{\X + \s^\T\q-\bar V(\s^\T\q)})(\X+\s^\T\q-\bar V(\s^\T\q))^{-\eta}]} ,\\
\bar\f^\q(\s) &= \frac{\EP[\q(\X + \s^\T \q - \bar V(\s^\T \q))^{-\eta}]}{\EP[(\X + \s^\T \q - \bar V(\s^\T \q))^{-\eta}]}.
\end{align*}
The induced price cross-impacts implies that there may exist complicated dependencies between prices of (statistically) independent assets.  The cross-impacts for an i.i.d.\ bivariate lognormal setting are displayed in Figure~\ref{fig:cross-impact}; due to the symmetry of the assets only the inverse demand functions for the first asset are provided.  Figure~\ref{fig:cross-impact} displays the contour lines for different price levels as a function of the joint liquidation amounts $\s = (s_1,s_2)$. If no cross-impacts existed, these contour lines would be vertical as this would imply $s_2$ does not impact the price of asset $1$.  Notably, with this interpretation in mind, there are less cross-impacts when $s_1$ is small, but the cross impacts can grow significantly as $s_1$ grows.
\begin{figure}
\centering
\begin{subfigure}[t]{0.45\textwidth}
\centering
\includegraphics[width=\textwidth]{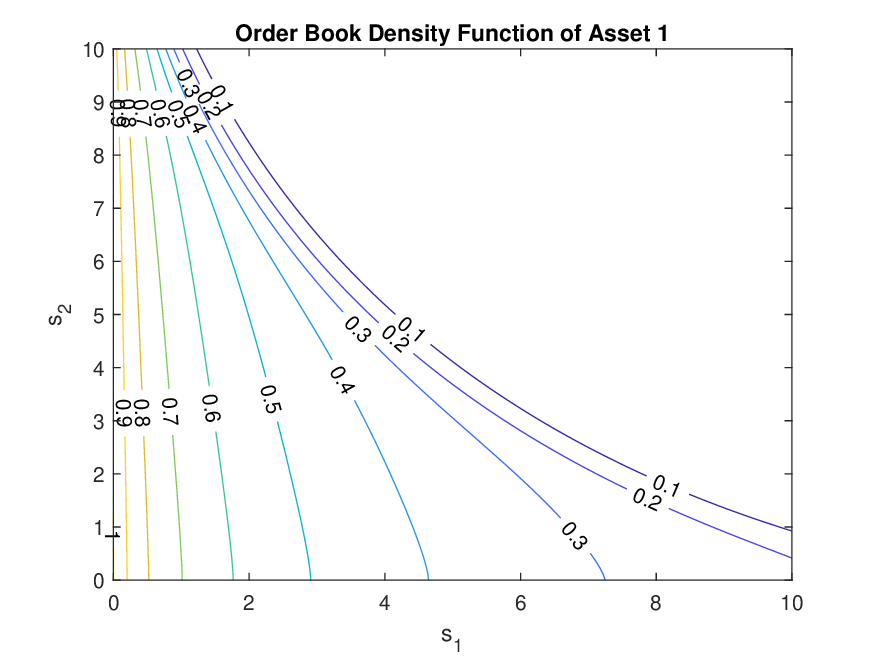}
\caption{Order book density $f_1^\q$ for asset $1$.}\label{fig:cross-impact-f}
\end{subfigure}
~
\begin{subfigure}[t]{0.45\textwidth}
\centering
\includegraphics[width=\textwidth]{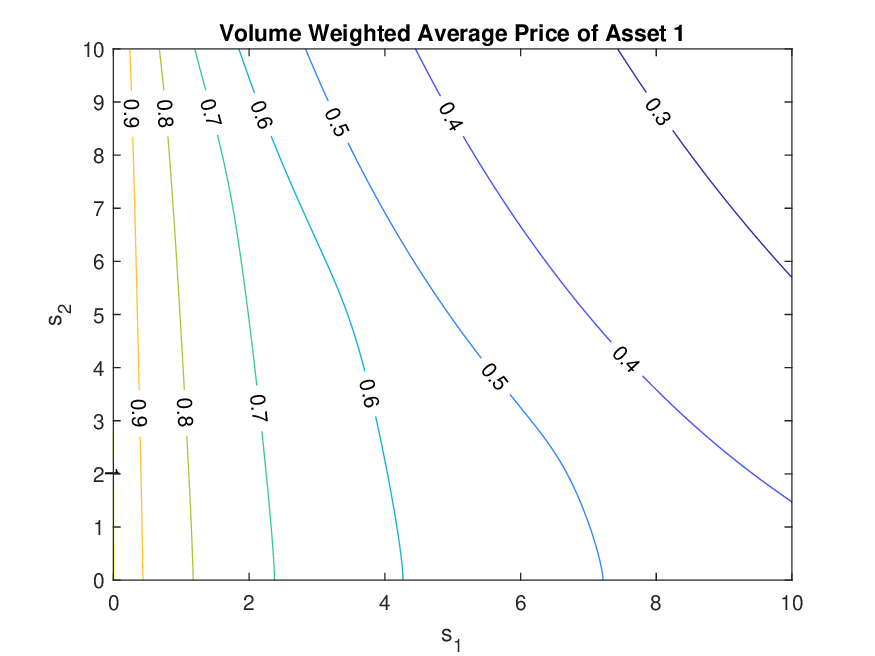}
\caption{Volume weighted average price $\bar f_1^\q$ for asset $1$.}\label{fig:cross-impact-fbar}
\end{subfigure}
\caption{Plot of the inverse demand functions $f_1^\q,\bar f_1^\q$ for the first asset of i.i.d.\ bivariate lognormal payoffs \blue{($Z_i \sim \operatorname{LogN}(-\frac{\sigma^2}{2},\sigma^2)$ at $\sigma = 1$ for $i \in \{1,2\}$) over a financial system of log-utility maximizers ($\eta = 1$) with aggregate assets $\X = 5$}.}\label{fig:cross-impact}
\end{figure}

In addition to the properties satisfied by these inverse demand functions, we also wish to note that, e.g., $\s \mapsto \s^\T\bar\f^\q(\s)$ is nondecreasing and concave for $\eta \in [0,1]$ due to Proposition~\ref{prop:V-power}. 
Further, as with the exponential utility function, neither inverse demand function is convex in general; herein this can be seen with a single asset $q$ at $s^* = \frac{\X\EP[(q-\essinf q)^{-\eta}]}{\EP[(q-\essinf q)^{1-\eta}]}$ for $q > \essinf q$ a.s.  In particular, as displayed in Figure~\ref{fig:lognormal-without}, neither $f^q$ nor $\bar f^q$ is convex at $s^* \approx 2.568$ for $q \sim \operatorname{LogN}(-\frac{\sigma^2}{2},\sigma^2)$ at $\sigma = 0.5$ with $\X = 2$ and $\eta = 1$. 
\blue{Furthermore, motivated by Corollary~\ref{cor:bernoulli}, in Figure~\ref{fig:lognormal-with} we plot the inverse demand functions for the same risk-\blue{sharing} system ($\X = 2$ and $\eta = 1$) for a liquidated lognormally distributed claim (i.e., $\sigma = 0.5$), but with a $0.001\%$ probability of systemic ruin (i.e., $p = 1 - 10^{-5}$ for independent Bernoulli distributed random variable $B[p]$).  Notably the order book density function $f^q$ is continuous with this inclusion of the risk of systemic ruin, though still neither $f^q$ nor $\bar f^q$ are convex mappings.}
\begin{figure}
\centering
\begin{subfigure}[t]{0.45\textwidth}
\centering
\includegraphics[width=\textwidth]{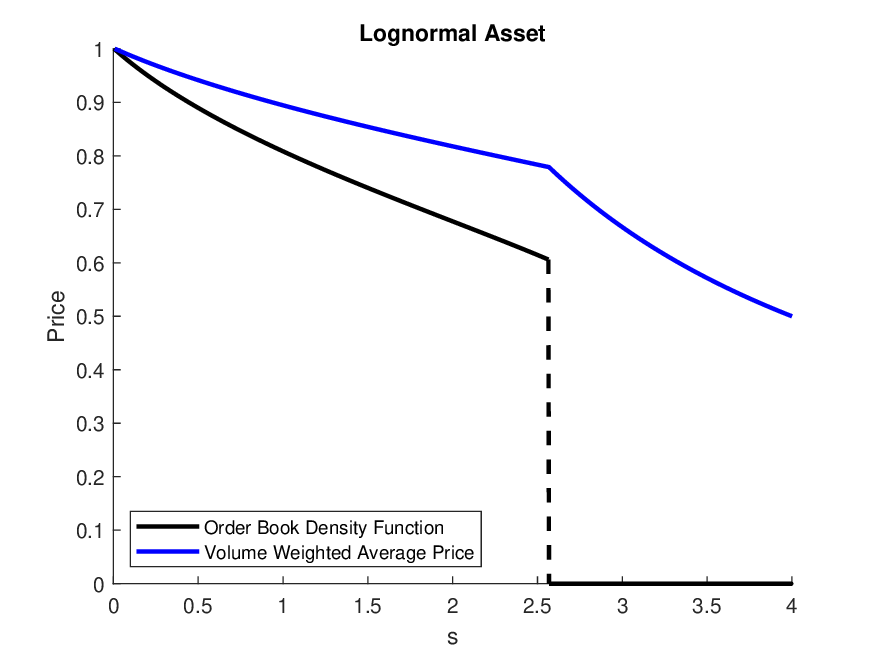}
\caption{Inverse demand functions for lognormal distribution.}\label{fig:lognormal-without}
\end{subfigure}
~
\begin{subfigure}[t]{0.45\textwidth}
\centering
\includegraphics[width=\textwidth]{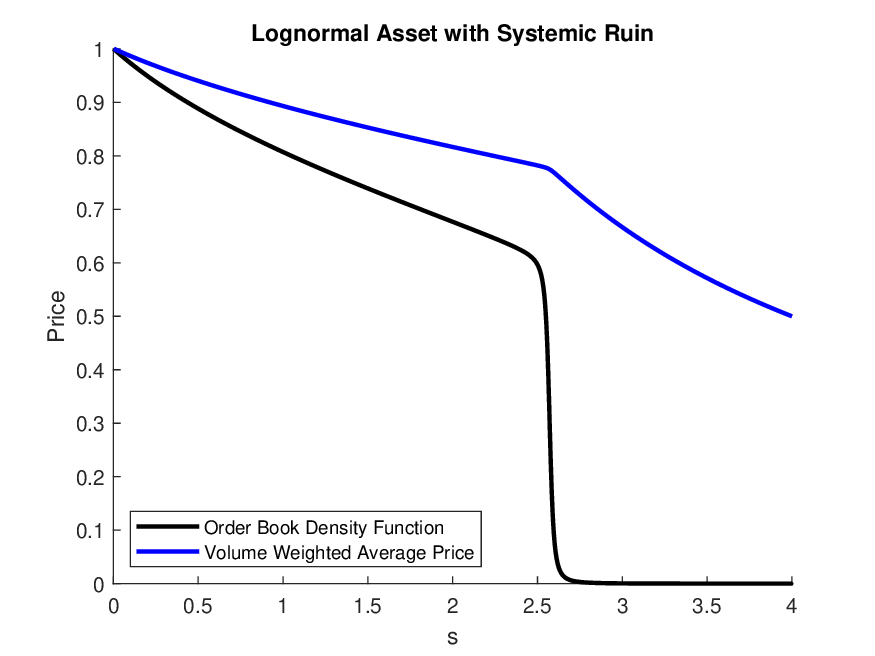}
\caption{Inverse demand functions for lognormal distribution with low probability of systemic ruin.}\label{fig:lognormal-with}
\end{subfigure}
\caption{Plot of the inverse demand functions $f^q,\bar f^q$ for an asset with lognormal payoffs \blue{($Z \sim \operatorname{LogN}(-\frac{\sigma^2}{2},\sigma^2)$ at $\sigma = 0.5$)} with 
\blue{(b)} \blue{and} without \blue{(a)} probability of systemic ruin \blue{over a financial system of log-utility maximizers ($\eta = 1$) with aggregate assets $\X = 2$}.}\label{fig:lognormal}
\end{figure}

We wish to conclude this \blue{section} with a few quick comments on the difficulties inherent in finding analytical forms for the pricing function $\bar V$.  Though numerical computation of $\bar V(Z)$ is straightforward through Monte Carlo simulation, analytical construction is hampered by the need for non-integer moments of a constant plus $Z$.  The combination of these requirements on the distribution generally make an explicit representation of the fixed point of $V$ intractable.   

%

\section{Conclusion}\label{sec:conclusion}
In this work we have introduced an equilibrium model for pricing externally liquidated assets.  This \blue{model} builds upon the seminal work by B\"uhlmann~\cite{buhlmann1980economic,buhlmann1984general} to endogenize the price impacts to find a clearing price in a financial market.  This \blue{approach} is in contrast to the typical approach in the fire sale literature in which inverse demand functions are exogenously given. In order to study these endogenous inverse demand functions we prove the existence and uniqueness of the \blue{the price of the liquidated assets}.  We additionally find that the resulting pricing functions satisfy the axioms of monetary risk measures; further study of this class of risk measures is left for future research.  Utilizing these results, we analyze two special cases -- all \blue{market participants} are exponential or power utility maximizers -- to study analytical structures.  The exponential utility setup provides a direct connection with the Esscher transform and provides analytical structure to the inverse demand functions whereas the power utility setup satisfies useful mathematical properties for any claim being liquidated \blue{(e.g., the monotonicity and concavity of the value of liquidated claims)}.  Importantly, we find an example -- the power utility setting -- in which these inverse demand functions generate price cross-impacts even for statistically independent assets.

\bibliographystyle{plain}
\small{\bibliography{bibtex2}}

\newpage
\setcounter{page}{1}
\appendix
\section{Construction of the (modified) B\"uhlmann equilibrium}\label{sec:buhlmann-appendix}
In Section~\ref{sec:buhlmann}, we introduced the basic details of the B\"uhlmann equilibrium.  For completeness, in this section we wish to present the general construction of the B\"uhlmann equilibrium $(Y,\Q)$.  The arguments presented here follow directly from~\cite{buhlmann1984general}.  

First, recall the setting of the modified B\"uhlmann equilibrium problem.  That is, consider a market of $n$ participants with utility functions $u_i$ and endowments $X_i$ into which some \emph{external} portfolio $Z \in L^\infty$ is sold.  The modified B\"uhlmann equilibrium $(Y,\Q)$ satisfies:
\begin{enumerate}
\item \textbf{Utility maximizing}: $Y_i \in \argmax_{\hat Y_i} \EP\left[u_i\left(X_i + \hat Y_i - \E^\Q[\hat Y_i]\right)\right]$ \blue{with $\EP\left[u_i\left(X_i +  Y_i - \E^\Q[ Y_i]\right)\right]\in\R$} for every $i \in \{1,2,...,n\}$
; and
\item \textbf{Equilibrium transfers}: $\sum_{i = 1}^n Y_i = Z$ for \emph{externally} sold position $Z \in L^\infty$.
\end{enumerate}
If $Z \equiv 0$ then this modified equilibrium coincides exactly with the typical B\"uhlmann equilibrium.
\blue{Recall that for simplicity of exposition, we have assumed that $\X +Z -\E^\Q[Z] \in \D$, and as an ansatz let 
\begin{align}
X_i +  Y_i - \E^\Q[ Y_i]\in\D,~i \in \{1,2,...,n\}.
\label{eq:ansatz}
\end{align}
}
%
The first order condition for utility maximizing implies
\begin{equation}\label{eq:foc}
u_i'(X_i + Y_i - \E^\Q[Y_i]) = \underbrace{\EP\left[u_i'(X_i + Y_i - \E^\Q[Y_i])\right]}_{=: C_i \in \Rplus} \frac{d\Q}{d\P} \quad \text{a.s.}
\end{equation}
By~\cite{borch1960,borch1962}, any equilibrium must, additionally, depend on $\omega \in \Omega$ only through $\gamma:=\X + Z - \E^\Q[Z] = \sum_{i = 1}^n (X_i + Y_i - \E^\Q[Y_i])$.  We must also have that $X_i + Y_i - \E^\Q[Y_i]$ depends on $\om$ only through $\gamma$.   Letting $\Y_i(\gamma) := X_i + Y_i - \E^\Q[Y_i]$ for every $i$ with $\sum_{i = 1}^n \Y_i(\gamma) = \gamma$ and $\phi(\gamma)$ is the associated Radon-Nikodym derivative of $\Q$ w.r.t.\ $\P$.
Therefore~\eqref{eq:foc} can be rewritten as $u_i'(\Y_i(\gamma)) = C_i \phi(\gamma)$.  Taking the derivative (w.r.t.\ $\gamma$) of the logarithm of both sides leads to the relation
\begin{equation}\label{eq:foc-2}
\underbrace{\frac{u_i''(\Y_i(\gamma))}{u_i'(\Y_i(\gamma))}}_{= -\rho_i(\Y_i(\gamma))} \Y_i'(\gamma) = \frac{\phi'(\gamma)}{\phi(\gamma)} \quad \forall i = 1,...,n
\end{equation}
where $\rho_i$ denotes the risk aversion of bank $i$.
As $\sum_{i = 1}^n \Y_i'(\gamma) = 1$ by construction, we recover the relation
\begin{equation}\label{eq:buhlmann-2}
1 = -\frac{\phi'(\gamma)}{\phi(\gamma)}\underbrace{\sum_{i = 1}^n \frac{1}{\rho_i(\Y_i(\gamma))}}_{=: n \rho(\gamma)^{-1}}
\end{equation}
for harmonic average $\rho$ of risk aversions $\rho_i,~1\leq i \leq n$.

Directly from~\eqref{eq:buhlmann-2}, the B\"uhlmann equilibrium $(Y,\Q)$ can be constructed.  As we are primarily concerned with the pricing measure $\Q$ within this work, we will first focus on that measure through the construction of $\phi(\gamma)$.  Specifically, $\phi'(\gamma) = -\frac{1}{n}\rho(\gamma)\phi(\gamma)$ or, equivalently (noting that $\EP[\phi(\gamma)] = 1$) 
\[\phi(\gamma) = \frac{\exp\left(-\frac{1}{n}\int_{c}^\gamma \rho(\xi)d\xi\right)}{\EP\left[\exp\left(-\frac{1}{n}\int_{c}^\gamma \rho(\xi)d\xi\right)\right]}\]
for arbitrary $c \in \D$.  The pricing measure $\Q$ is then constructed as in~\eqref{eq:Q} as $\frac{d\Q}{d\P} = \phi(\X + Z - \E^\Q[Z])$.
For this construction, implicitly we require $\Y_i(\gamma)$ so as to define the risk aversion $\rho$ of the harmonic representative agent.  Utilizing both~\eqref{eq:foc-2} and~\eqref{eq:buhlmann-2}, we recover the differential system
\begin{equation}\label{eq:ycal}
\Y_i'(\gamma) = \frac{1}{n}\frac{\rho(\gamma)}{\rho_i(\Y_i(\gamma))}
\end{equation}
with initial condition $\Y_i(c) \in \R$ such that $\E^\Q[\Y_i(\X+Z-\E^\Q[Z])] = \E^\Q[X_i]$ for every bank $i$.  We wish to note that~\cite{buhlmann1984general} studies the existence and uniqueness of the equilibrium by considering the existence and uniqueness of such an initial condition (with $Z \equiv 0$); such results are replicated within Theorem~\ref{thm:buhlmann}.

If the individual investments $Y_i$ were desired, then this can be found directly from the construction of $\Y_i(\X+Z-\E^\Q[Z])$ given the pricing measure $\Q$.
Specifically, 
\begin{equation*}
Y_i = -X_i + \Y_i(\X + Z - \E^\Q[Z]) + \lambda_i \E^\Q[Z]
\end{equation*}
for $\lambda_i \in \R$ arbitrary such that $\sum_{i = 1}^n \lambda_i = 1$.  By construction of $\Y(\X+Z-\E^\Q[Z])$, it immediately follows that equilibrium transfers $\sum_{i = 1}^n Y_i = Z$, \blue{and if $\bbd = \Rplus$ the ansatz \eqref{eq:ansatz} can be checked as in \cite{aase1993equilibrium}}.

\section{Proof from Section \ref{sec:buhlmann}}
\subsection{Proof of Proposition \ref{prop:representative}}

\begin{proof}
Note that $\rho_i: \D \to \R_+$ is nonnegative for every $i$ by the assumed properties of $u_i: \D \to \R$.  This implies $\rho: \D \to \R_+$ is nonnegative as well by construction in~\eqref{eq:rho}.  Therefore $\Y_i: \D \to \R$ is nondecreasing as its derivative is nonnegative (see~\eqref{eq:ycal}).
By construction, it immediately follows that
\begin{align*}
\rho(\gamma_1) &= n\left(\sum_{i = 1}^n \frac{1}{\rho_i(\Y_i(\gamma_1))}\right)^{-1} \geq n\left(\sum_{i = 1}^n \frac{1}{\rho_i(\Y_i(\gamma_2))}\right)^{-1} = \rho(\gamma_2)
\end{align*}
for $\gamma_1,\gamma_2 \in \D$ such that $\gamma_1 \leq \gamma_2$. 
\end{proof}

\section{Example of nonunique pricing}\label{sec:nonunique}
The extra conditions of Theorem~\ref{thm:unique} are sufficient (but not necessary) for the uniqueness of the equilibrium price.  We present the following counterexample to demonstrate that existence of an equilibrium price does not guarantee uniqueness and, thus, demonstrating some conditions are needed for such a property.  Notably, the following counterexample is under the setting $\bbd = \R$ which, by Theorem~\ref{thm:exists}, the existence of a fixed point for $H_Z$ is guaranteed.

\begin{example}\label{ex:nonunique}
Let $R: \R \to \R$ (i.e., $\bbd = \R$) be such that $R(z) = 1 - e^{-z+2.3}$ for any $z \in \R$.  Consider the discrete probability space $(\Omega := \{\omega_1,\omega_2\} \; , \; 2^\Omega \; , \; \P)$ such that $\P(\omega_1) = 0.01$ and $\P(\omega_2) = 0.99$.  Further, consider the setting in which the system-wide assets are provided by $\X(\omega_1) = 10^{-5}$ and $\X(\omega_2) = 100$, i.e., the financial system is shocked in scenario $\omega_1$.  Trivially $\X$ satisfies Assumption~\ref{ass:X} and $R$ satisfies assumption \ref{assump:R}, therefore this financial system has full domain ($\dom V = L^\infty$) by Theorem~\ref{thm:exists}.  
Consider the specific position to be liquidated: $Z \in L^\infty$ such that $Z(\omega_1) = 2$ and $Z(\omega_2) = 10^{-5}$.  Notably, this setting does \emph{not} satisfy any of the uniqueness conditions of Theorem~\ref{thm:unique}. In fact, there exist \emph{three} equilibrium prices for $Z$ all between $\essinf Z = 10^{-5}$ and $\esssup Z = 2$; up to rounding errors, the set of equilibrium prices of $Z$ is given by:  
\begin{align*}
V(Z) &= \{0.08403 \; , \; 1.38977 \; , \; 1.98985\}.
\end{align*}
The multiplicity of equilibria is visually clear in Figure~\ref{fig:nonunique} in which $H_Z(v)$ is plotted as a function of the initial valuation $v\in[10^{-5},2]$.
\begin{figure}
\centering
\includegraphics[width=0.5\textwidth]{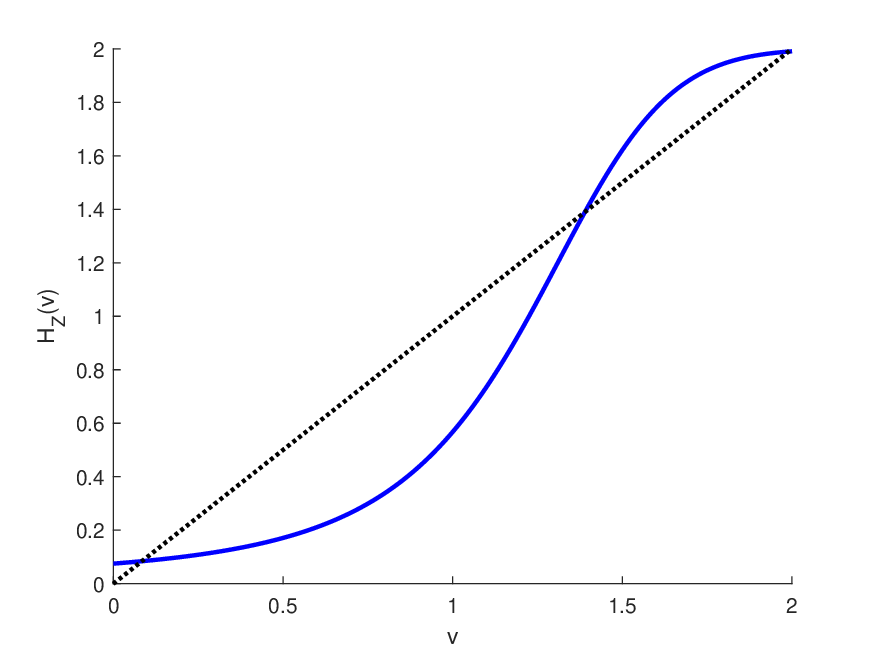}
\caption{Plot of the price $H_Z(v)$ resulting from initial valuation $v\in[10^{-5},2]$ for Example~\ref{ex:nonunique}.}\label{fig:nonunique}
\end{figure}
\end{example}

\section{Proofs from Section \ref{sec:value}}
\subsection{Proof of Theorem \ref{thm:exists}}
\begin{proof}
We first note that $R$ is continuous, and therefore so is $v\mapsto H_Z(v)$. 
\begin{enumerate}
\item If $\bbd =\R$, the existence of the fixed point now follows from Brouwer's fixed-point theorem, because 
$\essinf Z\le H_Z(v)\le\esssup Z$ for all $v\in\blue{\dom H_Z=}\R$.  In particular, these bounds hold for for $v\in[\essinf Z, \esssup Z].$

\item If $\bbd =\Rplus$, we still have that $\essinf Z\le H_Z(\essinf Z)$. \blue{If $\essinf[\X+Z]\not\in \dom H_Z$, and }
if, additionally, $\liminf\limits_{v \nearrow \essinf[\X+Z]} H_Z(v) =: L < \essinf[\X+Z]$, then fix $v^* \in (L , \essinf[\X+Z])$.  By assumption, there exists some $\bar v^* \in [v^* , \essinf[\X+Z])$ such that $H_Z(\bar v^*) < v^* \leq \bar v^*$.  Thus $H_Z(v)-v$ must have a root on $[\essinf Z , \bar v^*)$, which is a fixed point for $H_Z(v)$.
Alternatively, if \blue{$\essinf[\X+Z]\in \dom H_Z$ and if} $H_Z(\essinf[\X+Z]) \leq \essinf[\X+Z]$, then $H_Z(v)-v$ again must have a root on $[\essinf Z , \essinf[\X+Z]]$, which is again a fixed point for $H_Z(v)$.


\end{enumerate}
\blue{In any case, compactness of $V(Z)$ follows from: boundedness as $\essinf Z\le H_Z(v)\le\esssup Z\vee \essinf[\X+Z]$ for any $v \in \dom H_Z$; and closedness as the limit of fixed points of a continuous mapping is also a fixed point.}
\end{proof}

\subsection{Proof of Theorem \ref{thm:unique}}
\begin{proof}
We first consider conditions \eqref{thm:unique-comonotone} and \eqref{thm:unique-linear}.
Consider the derivative of $H_Z$ for a fixed external liquidation $Z$.  \blue{Let $v\in\dom H_Z$.} For simplicity of notation, define $R := R(\X+Z-v)$ and $R' := R'(\X+Z-v)$. 
\begin{align*}
H_Z'(v) &= \frac{\EP[\exp(-R)]\EP[Z R'\exp(-R)] - \EP[Z \exp(-R)]\EP[R' \exp(-R)]}{\EP[\exp(-R)]^2}.
\end{align*}
Recall, by construction, $R' > 0$.  Therefore, $H_Z'(v) \leq 0$ \blue{(the derivative is understood to be the right derivative, if $v$ is on the boundary of $\dom H_Z$)} 
if 
\[\EP[Z\exp(-R(\Zhat))]\EP[R'(\Zhat)\exp(-R(\Zhat))] - \EP[\exp(-R(\Zhat))]\EP[Z R'(\Zhat)\exp(-R(\Zhat))] \geq 0\]
where $\Zhat := \X+Z-v$.
Indeed, let $A(\omega) := \{\bar\omega \in \Omega \; | \; \Zhat(\bar\omega) \leq \Zhat(\omega)\}$. Then
\begin{align*}
&-H_Z'(v) \EP[\exp(-R)]^2=
\EP[Z\exp(-R(\Zhat))]\EP[R'(\Zhat)\exp(-R(\Zhat))] - \EP[\exp(-R(\Zhat))]\EP[Z R'(\Zhat)\exp(-R(\Zhat))]\\
&= \int_{\Omega} Z(\omega)\exp(-R(\Zhat(\omega)))\P(d\omega) \int_{\Omega} R'(\Zhat(\bar\omega))\exp(-R(\Zhat(\bar\omega)))\P(d\bar\omega)\\
&\qquad - \int_{\Omega} \exp(-R(\Zhat(\omega)))\P(d\omega) \int_{\Omega} Z(\bar\omega)R'(\Zhat(\bar\omega))\exp(-R(\Zhat(\bar\omega)))\P(d\bar\omega)\\
&= \int_{\Omega}\int_{\Omega} \exp(-R(\Zhat(\omega)))\exp(-R(\Zhat(\bar\omega))) \left[Z(\omega)-Z(\bar\omega)\right]R'(\Zhat(\bar\omega))\P(d\bar\omega)\P(d\omega)\\
&= \int_{\Omega}\int_{A(\omega)} \exp(-R(\Zhat(\omega)))\exp(-R(\Zhat(\bar\omega))) \left[Z(\omega)-Z(\bar\omega)\right]R'(\Zhat(\bar\omega))\P(d\bar\omega)\P(d\omega)\\
&\qquad + \int_{\Omega}\int_{A(\omega)^c} \exp(-R(\Zhat(\omega)))\exp(-R(\Zhat(\bar\omega))) \left[Z(\omega)-Z(\bar\omega)\right]R'(\Zhat(\bar\omega))\P(d\bar\omega)\P(d\omega)\\
&= \int_{\Omega}\int_{A(\omega)} \exp(-R(\Zhat(\omega)))\exp(-R(\Zhat(\bar\omega))) \left[Z(\omega)-Z(\bar\omega)\right]R'(\Zhat(\bar\omega))\P(d\bar\omega)\P(d\omega)\\
&\qquad + \int_{\Omega}\int_{A(\omega)} \exp(-R(\Zhat(\omega)))\exp(-R(\Zhat(\bar\omega))) \left[Z(\bar\omega)-Z(\omega)\right]R'(\Zhat(\omega))\P(d\bar\omega)\P(d\omega)\\
&= \int_{\Omega}\int_{A(\omega)} \underbrace{\exp(-R(\Zhat(\omega)))\exp(-R(\Zhat(\bar\omega)))}_{> 0} 
\left[Z(\omega)-Z(\bar\omega)\right]\left[R'(\Zhat(\bar\omega))-R'(\Zhat(\omega))\right]\P(d\bar\omega)\P(d\omega).
\end{align*}
Therefore, \blue{for $v\in\dom H_Z$,} $H_Z'(v) \le0$ if
$\left[Z(\omega)-Z(\bar\omega)\right]\left[R'(\Zhat(\bar\omega))-R'(\Zhat(\omega))\right]\ge0$.
This holds if either:
\begin{enumerate}[(a)]
\item $Z$ and $\X+Z$ are comonotonic because $R'$ is non-increasing;
\item $R$ is linear as $R'$ is a constant.
\end{enumerate}
In either case, we have that $\frac{\d}{\d v} (H_Z(v) -v)\leq -1 < 0$, which guarantees at most one root exists.

For cases \eqref{thm:unique-monotone} and \eqref{thm:unique-concave}, we note that $v = H_Z(v)$ if and only if $\Theta_Z(v) := \EP[(Z-v)\exp(-R(\X+Z-v))] = 0$. Additionally, by Proposition~\ref{prop:monotone-concave}, $z\mapsto z\exp(-R(\X+z))$ is almost surely non-decreasing or concave (under \eqref{thm:unique-monotone} and \eqref{thm:unique-concave} respectively) where it is well-defined.  First assume $\P(Z = \essinf Z) = 1$, i.e., $Z$ is a constant a.s.; then trivially the only root is given by $V(Z) = \essinf Z$.  Otherwise $\P(Z > \essinf Z) > 0$.  We will focus on $\frac{\d}{\d v} \Theta_Z(v)$ over feasible $v\in\R$, i.e.,
\begin{align*}
\frac{\d}{\d v}\Theta_Z(v) = -\EP\left[(1-(Z-v)R'(\X+Z-v))\exp(-R(\X+Z-v))\right].
\end{align*}
\begin{enumerate}[(a)]\setcounter{enumi}{2}
\item If $z \mapsto z\exp(-R(\X+z))$ is almost surely non-decreasing, then $z R'(\X + z) \leq 1, ~z\in\D$ a.s.. 
This guarantees $(Z-v)R'(\X+Z-v) \leq 1$ a.s.,\ because $(Z(\omega)-v)R'(\X(\omega)+Z(\omega)-v) \leq 0$ on $\{\omega \in \Omega \; | \; Z(\omega) \leq v\}$ for any feasible $v \in \R$.  As $\Theta_Z(\essinf Z) > 0$, then $\P(Z \leq v^*) > 0$ for any (feasible) price $v^* \in \R$ such that $\Theta_Z(v^*) = 0$.  As a direct consequence $\frac{\d}{\d v}\Theta_Z(v) < 0$ for any $v \geq v^*$ feasible which contradicts the existence of a multiplicity of equilibria.
\item If $z \mapsto z\exp(-R(\X+z))$ is almost surely concave, then $v \mapsto \Theta_Z(v)$ is concave as well.  As $\Theta_Z(\essinf Z) > 0$, the minimal (feasible) price $v^* \in \R$ such that $\Theta_Z(v^*) = 0$ (if it exists) must satisfy $\frac{\d}{\d v}\Theta_Z(v^*) < 0$.  By concavity, for any $v \geq v^*$ feasible must therefore also satisfy $\frac{\d}{\d v}\Theta_Z(v) < 0$ which contradicts the existence of a multiplicity of equilibria.
\end{enumerate}
\end{proof}

\begin{proposition}\label{prop:monotone-concave}
Assume $R$ satisfies Assumption~\ref{assump:R}.  If $z \in \Rplus \mapsto z \exp(-R(\X+z))$ is non-decreasing (concave) for fixed $\X \in \D$ then this mapping is non-decreasing (concave) over its entire domain $\D-\X$.
\end{proposition}
\begin{proof}
Consider the derivative of $z \exp(-R(\X+z))$ w.r.t.\ $z$, i.e.,
\[\frac{\d}{\d z} z \exp(-R(\X+z)) = \left[1 - z R'(\X+z)\right]\exp(-R(\X+z)).\]
Therefore, $\frac{\d}{\d z} z \exp(-R(\X+z)) > 0$ for any $z \leq 0$; and thus monotonicity holds over the entire domain if it holds for $z \in \Rplus$.
Further, by Assumption~\ref{assump:R}, this derivative is strictly decreasing on $z \leq 0$; therefore concavity holds over the entire domain if it holds for $z \in \Rplus$.
\end{proof}

\subsection{Proof of Corollary \ref{cor:exist-cond}}
\begin{proof}
We have already shown in Theorem \ref{thm:exists} that under the condition
$H_Z(\essinf [\X + Z]) \leq \essinf [\X + Z]$, when  \blue{$\essinf[\X+Z]\in\dom H_Z$,} or
$\liminf_{v \nearrow \essinf [\X + Z]} H_Z(v) < \essinf [\X + Z]$ if otherwise  \blue{$\essinf[\X+Z]\not\in\dom H_Z$,}  
we have existence of the fixed point. Given any of the conditions \eqref{thm:unique-comonotone}-\eqref{thm:unique-concave} of Theorem~\ref{thm:unique} we immediately recover uniqueness as well.  

We will show that this condition it is also necessary for existence under any of the conditions of Theorem~\ref{thm:unique}. 
Assume first that  \blue{$\essinf[\X+Z]\in\dom H_Z$,} and assume that $H_Z({\essinf [\X + Z]}) > {\essinf [\X+Z]}$.  We will demonstrate that this implies $H_Z(v)-v$ (equivalently $\Theta_Z(v) := \EP[(Z-v)\exp(-R(\X+Z-v))]$) has no roots on $[\essinf Z, \essinf [\X+Z] ]$, and therefore no fixed point for $H_Z(v)$ exists. 
Indeed, if either \eqref{thm:unique-comonotone} or \eqref{thm:unique-linear} of Theorem~\ref{thm:unique} hold then we have shown in the proof of Theorem \ref{thm:unique} that $H_Z'(v) \le0$, and therefore
$H_Z(v) > \essinf [\X+Z] \ge v$ for $v\in[\essinf Z, \essinf [\X+Z] ].$ 
Similarly, if \eqref{thm:unique-monotone} or \eqref{thm:unique-concave} of Theorem~\ref{thm:unique} hold then we have shown in the proof of Theorem \ref{thm:unique} that $\Theta_Z(v) \leq 0$ for every $v \geq V(Z)$.  As $H_Z(\essinf[\X+Z]) - \essinf[\X+Z] > 0$, $\Theta_Z(v) > 0$ for every $v \in[\essinf Z,\essinf[\X+Z]]$.

The proof is similar if  \blue{$\essinf[\X+Z]\not\in\dom H_Z$}.
By either \eqref{thm:unique-comonotone} or \eqref{thm:unique-linear} of Theorem~\ref{thm:unique}, we have that $H_Z'(v) \le0$.  Therefore, $H_Z(v) \ge \liminf_{v \nearrow \essinf[\X+Z]} H_Z(v) \geq \essinf [\X + Z] > v$ for all $v < \essinf [\X+Z]$, \blue{$v\in\dom H_Z$}. 
Similarly and as above, under either \eqref{thm:unique-monotone} or \eqref{thm:unique-concave} of Theorem~\ref{thm:unique}, we have that $\Theta_Z(v) > 0$ for all $v < \essinf [\X+Z]$.
In either setting we get that $H_Z(v) >v$ for all $v < \essinf [\X+Z]$, \blue{$v\in\dom H_Z$}, 
and thus no fixed point exists.
%
%
%

\end{proof}

\subsection{Proof of Corollary \ref{cor:Rplus}}
\begin{proof}
%
%
%
Recall $\essinf\X\in\D$ from Assumption \ref{ass:X}. 
To simplify notation throughout this proof, define $\Zhat(v) := \X+Z-v$ and $\tilde Z := \Zhat(\essinf[\X+Z]) = \X+Z-\essinf[\X+Z]$. 
Fix $0<\eps<\essinf\X/2$. By construction, $\EP[\exp(-R(\Zhat(v)))] < \infty$ for $v < \essinf[\X+Z]$ and $\EP[\exp(-R(\tilde Z))\ind_{\{\tilde Z \geq \eps\}}] < \infty$; therefore $\EP[\exp(-R(\tilde Z))\ind_{\{0\leq \tilde Z < \eps\}}] = \infty$ following from the assumption of the corollary.  Note that $Z - \essinf[\X+Z] < \eps-\X < -\eps$ a.s.\ on $\{0\leq \tilde Z < \eps\}$.
Therefore,
\begin{align}
\nonumber&\liminf_{v \nearrow \essinf[\X+Z]} H_Z(v)  = \liminf_{v \nearrow \essinf[\X+Z]} \frac{\EP[Z\exp(-R(\Zhat(v) ))]}{\EP[\exp(-R(\Zhat(v) ))]} \\
\nonumber&\quad=  \liminf_{v \nearrow \essinf[\X+Z]} \frac{\EP[( Z - \essinf [\X+Z]) \exp(-R(\Zhat(v) ))]}{\EP[\exp(-R(\Zhat(v) ))]} +\essinf [\X+Z]\\
\nonumber&\quad=  \liminf_{v \nearrow \essinf[\X+Z]}  \frac{\EP\left[ ( Z - \essinf [\X+Z])\exp\left(-R( \Zhat(v))\right)\left( \ind_{\{0\le \tilde Z < \eps\}} + \ind_{\{ \tilde Z \ge \eps\}}\right)\right]} {\EP\left[ \exp\left(-R( \Zhat(v))\right)\right]} +\essinf [\X+Z]\\
\nonumber&\quad=  \liminf_{v \nearrow \essinf[\X+Z]}  \frac{\EP\left[ ( Z - \essinf [\X+Z])\exp\left(-R( \Zhat(v))\right) \ind_{\{0\le \tilde Z < \eps\}} \right]} {\EP\left[ \exp\left(-R( \Zhat(v))\right)\right]} \\
\nonumber&\qquad+  \liminf_{v \nearrow \essinf[\X+Z]}  \frac{\EP\left[ ( Z - \essinf [\X+Z])\exp\left(-R( \Zhat(v))\right) \ind_{\{ \tilde Z \ge \eps\}}\right]} {\EP\left[ \exp\left(-R( \Zhat(v))\right)\right]} +\essinf [\X+Z]\\
\label{eq:proof:cor:Rplus}&\quad=  \liminf_{v \nearrow \essinf[\X+Z]} \frac{\EP\left[ ( Z - \essinf [\X+Z])\exp\left(-R( \Zhat(v))\right) \ind_{\{0\le \tilde Z < \eps\}} \right]} {\EP\left[ \exp\left(-R( \Zhat(v))\right)\right]} +\essinf [\X+Z] \\
\nonumber&\quad\le \liminf_{v \nearrow \essinf[\X+Z]} \frac{-\frac{\eps}{2} \EP\left[ \exp\left(-R( \Zhat(v))\right) \ind_{\{0\le \tilde Z < \eps\}} \right]} {\EP\left[ \exp\left(-R( \Zhat(v))\right)\ind_{\{0\le \tilde Z < \eps\}}\right]} + \essinf[\X+Z]\\
\nonumber&\quad= -\frac{\eps}{2} + \essinf[\X+Z] < \essinf [\X+Z].
\end{align}
We recover the equality in \eqref{eq:proof:cor:Rplus} as \[\lim_{v \nearrow \essinf[\X+Z]}  \frac{\EP\left[ ( Z - \essinf [\X+Z])\exp\left(-R( \Zhat(v))\right) \ind_{\{ \tilde Z \ge \eps\}}\right]} {\EP\left[ \exp\left(-R( \Zhat(v))\right)\right]} = 0\]
by assumption of this corollary.
\end{proof}

\section{Proofs from Section \ref{sec:extension}}

Within the proof of Lemma~\ref{lemma:barV} and later, we frequently make use of an optimization based representation for $\bar V$. 
\begin{proposition}\label{prop:opt-rep}
Assume $R$ satisfies Assumption~\ref{assump:R}.  
If any of the conditions of Theorem~\ref{thm:unique} holds, then the extension $\bar V$, can equivalently be formulated as:
\begin{align*}
\bar V(Z) &= \sup\{v \geq \essinf Z \; | \; H_Z(v) \geq v, \; \essinf[\X+Z] - v \in \bbd\}\\
    &= \sup\{v \geq \essinf Z \; | \; \EP[(Z-v)\exp(-R(\X+Z-v))] \geq 0, \; \essinf[\X+Z] - v \in \bbd\}
\end{align*}
for any $Z \in L^\infty$ for any choice of $\bbd$.
\end{proposition}
\begin{proof}
We wish to note that the two optimization problems are trivially equivalent by construction of $H_Z(v)$.  Further, the conditions of Theorem~\ref{thm:unique} imply: $v \mapsto H_Z(v)$ is monotonic under the first two conditions, $v \mapsto \EP[(Z-v)\exp(-R(\X+Z-v))]$ is monotonic under the third condition and monotonic for $v \geq V(Z)$ (for $Z \in \dom V$) in the fourth condition.

Fix $Z \in L^\infty$.  
If $Z \in \dom V$, i.e., $\#V(Z) = 1$ which we will treat as a scalar value, then we will prove the result by showing that $V(Z)$ is feasible for this optimization problem and $v > V(Z)$ is \emph{not} feasible.  As $V(Z) \geq \essinf Z$ by Lemma~\ref{lemma:barV}\eqref{lemma:barV-bound}, feasibility of $V(Z)$ follows from the construction of the equilibrium pricing problem as $\essinf[\X+Z]-V(Z) \in \bbd$ and $H_Z(V(Z)) = V(Z)$.  Further, the conditions of Theorem~\ref{thm:unique} imply either $H_Z(v) < v$ (for the first two conditions) or $\EP[(Z-v)\exp(-R(\X+Z-v))] < 0$ (for the latter two conditions) for any $v > V(Z)$.  This proves the result. 
If $Z \not\in\dom V$ then, noting that this implies $\bbd = \Rplus$, $H_Z(v) > v$ for every $v \in [\essinf Z,\essinf[\X+Z])$ by $H_Z(\essinf Z) > \essinf Z$ and the nonexistence of an equilibrium price.  Therefore it follows that $\sup\{v \geq \essinf Z \; | \; H_Z(v) \geq v, \; \essinf[\X+Z]-v\in\bbd\} = \essinf[\X+Z] = \bar V(Z)$ and the proof is complete.
\end{proof}

\subsection{Proof of Lemma~\ref{lemma:barV}}

\begin{proof}
\begin{enumerate}
\item First, if $Z \in \dom V$, then for any $v \in V(Z)$ it follows that $v = \EP[Z\exp(-R(\X+Z-v))]/\EP[\exp(-R(\X+Z-v))] \in [\essinf Z,\esssup Z]$ by construction; as such this holds for the minimal price $\bar V(Z)$ as well.  Assume, now, $Z \not\in \dom V$ and, as such, $\bbd = \Rplus$.  Immediately the lower bound $\bar V(Z) = \essinf[\X+Z] > \essinf Z$ holds.  If $\essinf[\X+Z] \leq \esssup Z$ then the upper bound holds; assume $\essinf[\X+Z] > \esssup Z$, then $Z \in \dom V$ because $H_Z(\esssup Z) \leq \esssup Z$ by construction and $\esssup Z$ being a feasible price for $Z$ which forms a contradiction and the proof is complete.
\item This follows immediately by the law invariance of $H_Z(v;\X)$ in $(Z,\X+Z)$. 
\item Note that $Z+z \in \dom V$ if and only if $Z \in \dom V$.  Furthermore, note that by construction $H_{Z+z}(v) = H_Z(v-z)+z$.  First, let $Z \in \dom V$ and $v \in V(Z+z)$.  Therefore $v = H_{Z+z}(v) = H_Z(v-z)+z$.  Take the ansatz that $v-z =: v^* \in V(Z)$, then $v^* = H_Z(v^*)$.  As this satisfies the fixed point problem and $Z \in \dom V$ then $V(Z+z) \subseteq V(Z)+z$.  The converse relation follows comparably, which immediately leads to the conclusion that $\bar V(Z+z) = \bar V(Z)+z$.  If $Z \not\in \dom V$ then $\bar V(Z+z) = \essinf[\X+Z+z] = \essinf[\X+Z] + z = \bar V(Z)+z$.  

\item \begin{enumerate}
    \item Consider $\bbd = \R$ where $\dom V = L^\infty$ by Theorem~\ref{thm:exists}.  Fix $\bar Z \in L^\infty$ and let $\ncal \subseteq \{Z \in L^\infty \; | \; \|Z - \bar Z\|_{\infty} \leq \delta\}$ for some $\delta > 0$ be a closed neighborhood of $\bar Z$.  Define $\bar H: \ncal \times [\essinf\bar Z - \delta , \esssup\bar Z + \delta] \to [\essinf\bar Z - \delta , \esssup\bar Z + \delta]$ be defined as the restriction of $H$, i.e., $\bar H(Z,v) := H_Z(v)$ for any $(Z,v) \in \ncal \times [\essinf\bar Z - \delta , \esssup\bar Z + \delta]$.  (Note that $H_Z(v) \in [\essinf Z , \esssup Z] \subseteq [\essinf\bar Z - \delta,\esssup\bar Z + \delta]$ by construction of $H$ and $\ncal$.)  Therefore, by construction, $V(Z) = \FIX_v \bar H(Z,v)$ for every $Z \in \ncal$.  By~\cite[Lemma C.1]{feinstein2020nash}, $Z \in \ncal \to V(Z)$ is a set-valued upper continuous mapping (i.e., continuous in the upper Vietoris topology).  As a direct consequence of \cite[Lemma 17.30]{AB07}, $Z \in \ncal \mapsto \bar V(Z) = \min V(Z)$ is lower semicontinuous.  As this is true for any (closed) neighborhood around any $\bar Z \in L^\infty$, $\bar V$ must be lower semicontinuous on the entire space $L^\infty$.
Furthermore, if any of the conditions of Theorem~\ref{thm:unique} hold, then uniqueness of the equilibrium price guarantees (scalar) continuity of $V$ which completes the proof in this setting.

    \item Consider $\bbd = \Rplus$.  Fix $\bar Z \in \dom V$ and let $\ncal \subseteq \{Z \in L^\infty \; | \; \|Z - \bar Z\|_{\infty} \leq \delta\}$ for some $\delta > 0$ be a closed neighborhood of $\bar Z$.  Define $\bar H: \ncal \times [\essinf\bar Z - \delta , \esssup\bar Z + \delta] \to 2^{[\essinf\bar Z - \delta , \esssup\bar Z + \delta]}$ be the set valued mapping with graph
\[\operatorname{graph}\bar H := \cl\left\{(Z,v,H_Z(v)) \; | \; Z \in \ncal, \, v \in [\essinf\bar Z - \delta , \esssup\bar Z + \delta], \, \essinf[\X+Z]-v\in\bbd\right\}.\]
(Note that $\bar H(Z,v) = \{H_Z(v)\}$ and $H_Z(v) \in [\essinf Z , \esssup Z] \subseteq [\essinf\bar Z - \delta , \esssup\bar Z + \delta]$ for any $Z \in \ncal$ and $v \in [\essinf\bar Z - \delta , \esssup\bar Z + \delta]$ such that $\essinf[\X+Z]-v\in\bbd$ by construction and continuity of $H$ and using the fact that $\bbd$ is open.)  Define $Z \in \ncal \mapsto \tilde V(Z) := \FIX_v \bar H(Z,v)$; this is a set-valued upper continuous mapping by~\cite[Lemma C.1]{feinstein2020nash}.  As a direct consequence of \cite[Lemma 17.30]{AB07}, $Z \in \ncal \cap \dom\tilde V \mapsto \min\tilde V(Z)$ is lower semicontinuous.  With the convention $\min\emptyset = \infty$ and noting that $\ncal \cap \dom\tilde V$ is closed (by noting that the graph of $\bar V$ is closed due to upper continuity), we can extend this result insofar as $\min\tilde V(\cdot)$ is lower semicontinuous on $\ncal$.
Importantly, by construction, $\bar V(Z) = \min\{\min\tilde V(Z) , \essinf[\X+Z]\}$ for any $Z \in \ncal$ since $\bar V(Z) = \min\tilde V(Z)$ on $Z \in \ncal \cap \dom V$ and $\bar V(Z) \leq \min\tilde V(Z)$ on $Z \in \ncal \cap (\dom V)^c$ (with the convention $\min\emptyset = \infty$). Therefore, $Z \in \ncal \mapsto \bar V(Z)$ is lower semicontinuous as the minimum of a lower semicontinuous and a continuous mapping.  As this is true for any (closed) neighborhood around any $\bar Z \in L^\infty$, $\bar V$ must be lower semicontinuous on the entire space $L^\infty$.

It remains to show that $\bar V$ is continuous if any of the conditions of Theorem~\ref{thm:unique} hold.
    By Proposition~\ref{prop:opt-rep} (modified by translation by $\essinf Z$), we can utilize an optimization representation of $\bar V$.
    By the Berge maximum theorem, if
\begin{align}
&Z \mapsto D(Z) \\
&:= \cl\{v \in [0,\essinf[\X+Z]-\essinf Z) \; | \; \EP[(Z-\essinf Z-v)\e{-R(\X+Z-\essinf Z-v)} ] \geq 0\}
\end{align}
is a set-valued continuous mapping (in the Vietoris topology) then the result holds because $Z \mapsto \essinf Z$ is continuous in the strong topology.  By the closed graph theorem (see, e.g., \cite[Theorem 17.11]{AB07}) and (almost sure) continuity of $z \mapsto z\exp(-R(\X+z))$, $D$ is an upper continuous mapping.  Now, let $\mathcal{V} \subseteq \R_+$ be open in the subspace topology and define $D^-[\mathcal{V}] := \{Z \in L^\infty \; | \; D(Z) \cap \mathcal{V} \neq \emptyset\}$; if $D^-[\mathcal{V}]$ is open then $D$ is lower continuous.  Let $Z \in D^-[\mathcal{V}]$ and, in particular, let $v \in \mathcal{V}$ such that $v \in D(Z)$.
    \begin{enumerate}
    \item If $H_{Z-\essinf Z}(v)-v > 0$ then there exists a neighborhood $\mathcal{N}_Z$ around $Z$ such that $\mathcal{N}_Z \subseteq D^-[\mathcal{V}]$ by continuity of $Z \mapsto H_{Z-\essinf Z}(v)-v$.
    \item If $H_{Z-\essinf Z}(v)-v = 0$ then:
        \begin{enumerate}
        \item If $v > 0$ then, by $\mathcal{V}$ open, take $\epsilon > 0$ so that $v-\epsilon \in \mathcal{V}$.  If Theorem~\ref{thm:unique}\eqref{thm:unique-comonotone}-\eqref{thm:unique-linear} holds then $\hat v \mapsto H_{Z-\essinf Z}(\hat v) - \hat v$ is strictly decreasing; if Theorem~\ref{thm:unique}\eqref{thm:unique-monotone}-\eqref{thm:unique-concave} holds then $\hat v \mapsto \EP[(Z-\essinf Z - \hat v)\exp(-R(\X+Z-\essinf Z-\hat v))]$ is strictly decreasing (in a neighborhood of $v$).  Therefore for $\epsilon$ small, $H_{Z-\essinf Z}(v-\epsilon) - (v-\epsilon) > 0$ and the result follows by continuity as in the prior case.
        \item If $v = 0$ then, by construction, $D^-[\mathcal{V}] = L^\infty$ and the result follows.
        \end{enumerate}
    \end{enumerate}
\end{enumerate}

\item First, we wish to note that the condition imposed herein appears also in Theorem~\ref{thm:unique} and, therefore, the optimization representation of Proposition~\ref{prop:opt-rep} holds. 
    \begin{enumerate}
    \item Let $Z_1 \geq Z_2$.  Note that this implies $\essinf Z_1 \geq \essinf Z_2$ and $\essinf[\X+Z_1] \geq \essinf[\X+Z_2]$.  If $\bar V(Z_2) \leq \essinf Z_1$ then, by the lower bound on $\bar V(Z_1)$ proven above, monotonicity trivially follows.  Assume, now, that $\bar V(Z_2) > \essinf Z_1$.  By a straightforward application of Proposition~\ref{prop:opt-rep}, $\bar V(Z_1) \geq \bar V(Z_2)$ if, and only if, 
\begin{align}
&\{v \geq \essinf Z_1 \; | \; \EP[(Z_1-v)\exp(-R(\X+Z_1-v))] \geq 0, \; \essinf [\X+Z_1] - v \in \bbd\} \\
&\supseteq \{v \geq \essinf Z_2 \; | \; \EP[(Z_2-v)\exp(-R(\X+Z_2-v))] \geq 0, \; \essinf [\X+Z_2]- v \in \bbd\}.
\end{align}
By Proposition~\ref{prop:opt-rep} and the assumption that $\bar V(Z_2) > \essinf Z_1$, these sets are non-empty. By assumption and Proposition~\ref{prop:monotone-concave}, $\EP[(Z_1-v)\exp(-R(\X+Z_1-v))] \geq \EP[(Z_2-v)\exp(-R(\X+Z_2-v))]$ and $\essinf [\X+Z_1]- v \geq \essinf [\X+Z_2] - v$.  Therefore the constraints are more restrictive w.r.t.\ $Z_2$ than $Z_1$ and the result follows.

    \item By translativity and monotonicity, for any $Z_1,Z_2 \in L^\infty$, $\bar V(Z_1) \leq \bar V(Z_2 + \|Z_1 - Z_2\|_\infty) = \bar V(Z_2) + \|Z_1 - Z_2\|_\infty$.  Taking the same inequality but switching $Z_1$ and $Z_2$ proves $|\bar V(Z_1) - \bar V(Z_2)| \leq \|Z_1 - Z_2\|_\infty$. 
    \end{enumerate}

\item First, we wish to note that the condition imposed herein appears also in Theorem~\ref{thm:unique} and, therefore, the optimization representation of Proposition~\ref{prop:opt-rep} holds. 
    \begin{enumerate} 
    \item $\bar V$ is concave if its hypograph 
    \[\operatorname{hypo}\bar V := \{(Z,v) \in L^\infty \times \R \; | \; \EP[(Z-v)\exp(-R(\X+Z-v))] \geq 0, \; \essinf[\X+Z] - v \in \bbd\}\] 
    is convex.  Let $(Z_1,v_1),(Z_2,v_2) \in \operatorname{hypo}\bar V$ and $\lambda \in [0,1]$.  Note that $\essinf[\X + \lambda Z_1 + (1-\lambda)Z_2] - [\lambda v_1 + (1-\lambda)v_2] \geq \lambda [\essinf[\X+Z_1] - v_1] + (1-\lambda) [\essinf[\X+Z_2] - v_2]$. 
    Consider now $(Z,v) \mapsto \EP[(Z-v)\exp(-R(\X+Z-v))]$.  By the concavity assumption on $z \mapsto z\exp(-R(\X+z))$ (and Proposition~\ref{prop:monotone-concave}),
    \begin{align*}
    \EP[&(\lambda(Z_1-v_1)+(1-\lambda)(Z_2-v_2))\exp(-R(\X+\lambda(Z_1-v_1)+(1-\lambda)(Z_2-v_2)))]\\
    &\geq \lambda\EP[(Z_1-v_1)\exp(-R(\X+Z_1-v_1))] + (1-\lambda)\EP[(Z_2-v_2)\exp(-R(\X+Z_2-v_2))] \geq 0. 
    \end{align*} 
    From these properties it is trivial to conclude $\lambda(Z_1,v_1)+(1-\lambda)(Z_2,v_2) \in \operatorname{hypo}\bar V$ and the proof is concluded.

    \item $\bar V$ is weak* upper semicontinuous if and only if $\{Z \in L^\infty \; | \; \bar V(Z) \geq v\}$ is weak* closed for every $v \in \R$.  By~\cite[Proposition 5.5.1]{KS09} and the concavity of $\bar V$, this is true if and only if $\{Z \in L^\infty \; | \; \bar V(Z) \geq v, \; \|Z\|_\infty \leq k\}$ is closed in probability for every $v \in \R$ and $k \in \Rplus$.  Let $Z_n \to Z$ in probability so that $Z_n \in \{Z \in L^\infty \; | \; \bar V(Z) \geq v, \; \|Z\|_\infty \leq k\}$.  First, $\|Z\|_\infty \leq k$ trivially.  Now we wish to show that $\bar V(Z) \geq v$; we will accomplish this separately if $\bbd = \R$ and if $\bbd = \Rplus$.

    Let $\bbd = \R$.  If $v \leq \essinf Z$ then $\bar V(Z) \geq v$ trivially by Property~\eqref{lemma:barV-bound}.  Let $v > \essinf Z$ and define $\mathcal{D}_k := \{Z \in L^\infty \; | \; \|Z\|_\infty \leq k\}$.  By construction, $\bar V(Z) \geq v$ if and only if $\EP[(Z-v)\exp(-R(\X+Z-v))] \geq 0$.  As $Z \in \mathcal{D}_k \mapsto \EP[(Z-v)\exp(-R(\X+Z-v))]$ is continuous w.r.t.\ convergence in probability, the result follows.

    Let $\bbd = \Rplus$.  If $v \leq \essinf Z$ then $\bar V(Z) \geq v$ trivially by Property~\eqref{lemma:barV-bound}.  Let $v > \essinf Z$.  By construction, $\bar V(Z) \geq v$ if and only if $\EP[(Z-v+\epsilon)\exp(-R(\X+Z-v+\epsilon))] > 0$ for every $\epsilon > 0$ and $\X+Z-v \geq 0$ a.s.  (The ``if'' statement holds immediately.  To prove the ``only if'' claim: if $\X+Z-v > 0$ a.s.\ then, by continuity, $\EP[(Z-v)\exp(-R(\X+Z-v))] \geq 0$ as desired; if $\P(\X+Z-v = 0) > 0$ then $v = \essinf[\X+Z]$ (with $\X+Z$ attaining its essential infimum with positive probability) and -- as a consequence -- $\bar V(Z) \geq v$ only if $Z \not\in \dom V$, i.e., $H_Z(v-\epsilon) > v-\epsilon$ for every $\epsilon > 0$.)  As such, the existence of an almost surely converging subsequence implies $\X+Z-v \geq 0$ a.s.  It remains to show that $\EP[(Z-v+\epsilon)\exp(-R(\X+Z-v+\epsilon))] > 0$ for every $\epsilon > 0$.  By continuity of $R$, $\EP[(Z-v+\epsilon)\exp(-R(\X+Z-v+\epsilon))] \geq 0$ for every $\epsilon > 0$.  If there exists some $\epsilon^* > 0$ such that $\EP[(Z-v+\epsilon^*)\exp(-R(\X+Z-v+\epsilon^*))] = 0$ then, by the concavity assumption, $\EP[(Z-v+\epsilon)\exp(-R(\X+Z-v+\epsilon))] < 0$ for every $\epsilon \in (0,\epsilon^*)$ which forms a contradiction and the result follows.
    \end{enumerate}
\end{enumerate}
\end{proof}

\subsection{Proof of Corollary \ref{cor:bernoulli}}
\begin{proof}
To simplify this proof, we will denote $\X_p = B[p](\X-\essinf\X) +\essinf\X$ and $Z_p=B[p](Z-\essinf Z ) +\essinf Z$ for $p\in[0,1]$ throughout. 
First, we may assume that $\P(Z \ne Z_p) >0, ~p\in[0,1)$. Otherwise, $Z=c$ for some constant $c \in \R$, and $V(c)=\hat V(c)=c$.
Second, note that if $R$ satisfies conditions \eqref{thm:unique-linear}-\eqref{thm:unique-concave} of Theorem~\ref{thm:unique}, the uniqueness result holds for the pair $\X_p, Z_p$ for any $p \in [0,1]$.  If Theorem~\ref{thm:unique}\eqref{thm:unique-comonotone} holds for $\X,Z$ (i.e., $Z$ and $\X+Z$ are comonotonic), then the same holds also for $\X_p,Z_p$.
Third, note that by our assumption (and to keep $R$ satisfying Assumption \ref{assump:R}), $R(0)=-\infty$. 
Finally, note that for $0\leq p<1$: 
\begin{align*}
\essinf[\X_p+Z_p]&=\essinf[ B[p](\X+Z-\essinf\X-\essinf Z) +\essinf\X+\essinf Z ] \\
&= \essinf[ B[p](\X+Z-\essinf[\X+Z]) ] + \essinf[\X+Z] = \essinf[\X+Z] .
\end{align*}
Therefore, by Corollary \ref{cor:Rplus}, $Z_p\in\dom V(\cdot;\X_p)$ for any $0<p<1$ since 
\begin{align}
\nonumber &\EP\left[ \exp\left(-R( \X_p+Z_p - \essinf [\X_p+Z_p])\right)\right]\\
\nonumber &= \EP\left[ \exp\left(-R( B[p](\X+Z -\essinf [\X+Z] ) +\essinf[\X+Z] - \essinf[\X+Z])\right)\right]\\
&=\EP\left[ \exp\left(-R( B[p](\X+Z - \essinf [\X+Z] ))\right)\right]=\infty.\label{eq:cor:Rplus-ok}
\end{align} 

To complete this proof, first, for arbitrary $Z \in L^\infty$, we will prove $\lim_{p \nearrow 1} V(Z_p,\X_p)$ exists.  Then we will utilize the definitional representation~\eqref{eq:barV} of $\bar V$ to prove that $\hat V(Z) = \bar V(Z)$ for any $Z \in L^\infty$.

We will prove $\lim_{p \nearrow 1} V(Z_p;\X_p)$ exists by showing that $p \in (0,1) \mapsto V(Z_p;\X_p)$ is monotonic in $p$.  Fix $p \in (0,1)$ and $B[p]$ independent from $Z,\X$.  To simplify notation, let $V := V(Z_p;\X_p)$ and $V' := \frac{\d}{\d p} V(Z_p;\X_p)$.  Note that $V$ is the solution to the fixed point problem $V = H_{Z_p}(V;\X_p)$, i.e.,
\begin{align*}
V &= 
\frac{\EP[(Z-\essinf Z)\exp(-R(\X +  Z - V))]p}{\EP[\exp(-R(\X +  Z - V))]p + \exp(-R(\essinf[\X+Z] - V))(1-p)} +\essinf Z.
\end{align*}
Therefore, assuming the derivative $V'$ exists, it must satisfy: 
\begin{align*}
&V'\left(\EP[\exp(-R(\X +  Z - V))]p + \exp(-R(\essinf[\X+Z] - V))(1-p)\right)\\
&\qquad + (V - \essinf Z)\left(V'\EP[R'(\X+ Z-V)\exp(-R(\X+ Z-V))]p\right) \\
&\qquad+ (V - \essinf Z) \left(\EP[\exp(-R(\X+ Z-V))] -\exp(-R(\essinf[\X+Z]-V))  \right)\\
&\qquad+(V - \essinf Z) V'R'(\essinf[\X+Z]-V)\exp(-R(\essinf[\X+Z]-V))(1-p)\\
&\quad = V'\EP[ (Z-\essinf Z) R'(\X +  Z - V)\exp(-R(\X+ Z-V))]p + \EP[ (Z-\essinf Z)\exp(-R(\X+ Z-V))]\\
&\Rightarrow\; V'\Big(\EP[(1-( Z-V)R'(\X+ Z-V))\exp(-R(\X+ Z-V))]p \\
&\qquad+ (1 + (V-\essinf Z)R'(\essinf[\X+Z]-V))\exp(-R(\essinf[\X+Z]-V))(1-p)\Big)\\
&\qquad = \EP[( Z-V )\exp(-R(\X+ Z-V))] + (V-\essinf Z)\exp(-R(\essinf[\X+Z]-V))\\
&\Rightarrow\; V'\EP[(1-(Z_p-V)R'(\X_p+Z_p-V))\exp(-R(\X_p+Z_p-V))]\\
&\qquad = \EP[ (Z-\essinf Z)\exp(-R(\X+ Z-V))] \\
&\qquad+ (V-\essinf Z)(\exp(-R(\essinf[\X+Z]-V)) - \EP[\exp(-R(\X+ Z-V))]).
\end{align*}
Therefore $V'$ exists if 
\begin{align}
\EP[(1-(Z_p-V)R'(\X_p+Z_p-V))\exp(-R(\X_p+Z_p-V))] \neq 0.
\label{eq:Theta-deriv-constr}
\end{align}
Under conditions~\eqref{thm:unique-monotone}-\eqref{thm:unique-concave} of Theorem~\ref{thm:unique},~\eqref{eq:Theta-deriv-constr} is satisfied by construction as in the proof of Theorem~\ref{thm:unique} as this expression is strictly positive.  In cases~\eqref{thm:unique-comonotone}-\eqref{thm:unique-linear} of Theorem~\ref{thm:unique}, we wish to rewrite \eqref{eq:Theta-deriv-constr}.  Specifically, by construction of $V$,
\begin{align*}
&\EP[(1-(Z_p-V)R'(\X_p+Z_p-V))\exp(-R(\X_p+Z_p-V))]\\
&\quad = \EP[\exp(-R(\X_p+Z_p-V))]\\
&\qquad+\frac{\EP[Z_p\exp(-R(\X_p+Z_p-V))]\EP[R'(\X_p+Z_p-V)\exp(-R(\X_p+Z_p-V))]}{\EP[\exp(-R(\X_p+Z_p-V))]}\\
&\qquad- \frac{\EP[\exp(-R(\X_p+Z_p-V))]\EP[Z_p R'(\X_p+Z_p-V)\exp(-R(\X_p+Z_p-V))]}{\EP[\exp(-R(\X_p+Z_p-V))]}.
\end{align*}
Therefore, as shown in the proof of Theorem~\ref{thm:unique}, this expression is strictly positive.

Therefore, $V' \geq 0$ if, and only if, $\EP[ (Z-\essinf Z)\exp(-R(\X+ Z-V))] + (V-\essinf Z)(\exp(-R(\essinf[\X+Z]-V)) - \EP[\exp(-R(\X+ Z-V))]) \geq 0$.  Since $ Z-\essinf Z \geq 0$ a.s.: $\EP[ (Z-\essinf Z)\exp(-R(\X+ Z-V))] \geq 0$ trivially, $V \geq 0$ by Proposition~\ref{lemma:barV}\eqref{lemma:barV-bound}, and $\exp(-R(\essinf[\X+Z]-V)) \geq \exp(-R(\X+ Z-V))$ by monotonicity of $R$.  Therefore $V' \geq 0$.  In particular, this implies that $p \in (0,1) \to  V(Z_p;\X_p)$ is nondecreasing.  Additionally, $V(Z_p;\X_p) \leq  \essinf [\X+Z]$ for every $p \in (0,1)$ by Corollary \ref{cor:exist-cond}. 
Therefore, by application of the monotone convergence theorem, we can guarantee the existence of $\hat V(Z)$.

We will complete this proof by proving that this limit is equivalent to the form~\eqref{eq:barV}.  Fix $Z \in L^\infty$.  First, we will show $\hat V(Z) = \bar V(Z)$ if $Z \in \dom V$.  Second, we will consider $Z \not\in \dom V$. 
\begin{enumerate}
\item First, fix $Z \in L^\infty$ such that $V(Z)$ exists. 
We will prove $\bar V(Z) = \lim_{p \nearrow 1} V(Z_p;\X_p) $ by showing that $p \in [0,1] \mapsto V(Z_p; \X_p ) $ is continuous. 
By Corollary~\ref{cor:Rplus} and choice of $Z$, $V(Z_p;\X_p) $ exists for any $p \in [0,1)$, as we have shown in \eqref{eq:cor:Rplus-ok}, and $V(Z;\X)$ exists by our assumption. By construction of $\bar V$, we have that $\bar V(Z_p;\X_p)= V(Z_p; \X_p),~p\in[0,1].$
By Lemma~\ref{lemma:barV}\eqref{lemma:barV-bound}, Corollary \ref{cor:exist-cond} and properties of the pricing function under $\bbd := \Rplus$, $V(Z_p;\X_p)  \in [\essinf Z,\min\{\essinf[\X+Z],\esssup Z\}]$ for any $p \in [0,1]$. For any probability $p \in [0,1]$: 
\begin{align}
V(Z_p;\X_p) &= \FIX_v \left\{\bar H(p,v)\right\}\\
\bar H(p,v) &:= \frac{\EP[Z\exp(-R(\X+Z-v))]p + \essinf Z \exp(-R(\essinf[\X+ Z]-v))}{\EP[\exp(-R(\X+Z-v)]p + \exp(-R(\essinf[\X+ Z]-v))(1-p)} \wedge \essinf [\X+Z].
\end{align}
As $\bar H: [0,1] \times [\essinf Z,\min\{\essinf[\X+Z],\esssup Z\}] \to [\essinf Z,\min\{\essinf[\X+Z],\esssup Z\}]$ is jointly continuous and $V(Z_p;\X_p)$ is unique for every $p \in [0,1]$, the fixed point mapping is continuous (see, e.g.~\cite[Lemma C.1]{feinstein2020nash}), i.e., $p \in [0,1] \to V(Z_p; \X_p ) $ is continuous.  
The result now follows as $\hat V(Z;\X) = \lim\limits_{p \nearrow 1} V(Z_p;\X_p) = V(Z;\X).$
%
\item Fix $Z \in L^\infty$ such that $Z \not\in \dom V$.
Consider first the case when 
\blue{$\essinf[\X+Z]\not\in\dom H_Z$}. Then we must have that $\liminf_{v \nearrow \essinf [\X + Z]} H_Z(v) \ge \essinf [\X + Z]$. Because $R$ is increasing, it also follows that $\EP[( Z - \essinf [\X + Z])\exp(-R(\X+ Z-v))] \ge0,$ for all $\essinf Z\le v < \essinf[\X+Z].$ Therefore, 
$
\EP[( Z - v)\exp(-R(\X+ Z-v))]  >0,~\essinf Z\le v < \essinf[\X+Z].
$
This also follows in case 
\blue{$\essinf[\X+Z]\in\dom H_Z$} (in which case we must have $H_Z( \essinf[\X+Z]) >  \essinf[\X+Z],$ 
and thus
$
\EP[( Z - v)\exp(-R(\X+ Z-v))]  >0,~\essinf Z\le v < \essinf[\X+Z],
$
as otherwise we would reach a contradiction to $Z \not\in \dom V$). For any $\epsilon\in\big(0,  \frac{\essinf\X}{2}\big),$
define 
\begin{align}
&p_\epsilon := \\
&\frac{(\essinf\X-\epsilon)\exp(-R(\epsilon))}{(\essinf\X-\epsilon)\exp(-R(\epsilon)) + \EP[(Z - \essinf[\X+Z] +\epsilon)\exp(-R(\X+ Z-\essinf[\X+Z]+\epsilon))]} \in (0,1)
\end{align}
so that $V(Z_{p_\epsilon};\X_{p_\epsilon}) = \essinf[\X+Z]-\epsilon.$  
Using the fact that $V'\ge0$, $V(Z_{p};\X_{p})  \ge  \essinf[\X+Z]-\epsilon$ for any $p>p_{\epsilon}$. Therefore, 
$p_{\epsilon_2} > p_{\epsilon_1}$ for $\epsilon_2 < \epsilon_1$, and
$\lim_{\epsilon \searrow 0} p_\epsilon$ exists and is bounded from above by $1$.  If $\lim_{\epsilon \searrow 0} p_\epsilon = 1$, then we conclude:
\[\lim_{p \nearrow 1} V(Z_p;\X_p) = \lim_{\epsilon \searrow 0} V(Z_{p_\epsilon};\X_{p_\epsilon})= \lim_{\epsilon \searrow 0} (\essinf[\X+Z]-\epsilon)=  \essinf [\X+Z].\] 
Otherwise, if $p_* := \lim_{\epsilon \searrow 0} p_\epsilon < 1$.  Then, by the above construction, $V(Z_p;\X_p) = \essinf[\X+Z]$ for every $p > p_*$ and, as a direct consequence, $\lim_{p \nearrow 1} V(Z_p;\X_p) = \essinf[\X+Z]$ as well.
\end{enumerate}
\end{proof}

\section{Proofs for Section \ref{sec:idf}}
\subsection{Proof of Lemma~\ref{lemma:f}}
\begin{proof}
\begin{enumerate}
\item First we will show that $f^q(s) \in \R$ for every $s \in \R_+$.  This is trivially true if $sq \not\in \dom V$.  Assume $sq \in \dom V$; $f^q(s)$ exists if and only if the denominator ($\EP[(1-[sq - \bar V(sq)]R'(\X+sq-\bar V(sq)))\exp(-R(\X+sq-\bar V(sq)))]$) is nonzero.  In fact, we will demonstrate that this denominator is strictly positive.  Note also that $\bar V(sq) = V(sq)$ by Theorem~\ref{thm:unique}.  
As demonstrated in the proof of Corollary~\ref{cor:bernoulli}, this denominator is strictly positive under any of the conditions of Theorem~\ref{thm:unique}.
%

Second assume $s \in \operatorname{int}\{s \in \R_+ \; | \; sq \in \dom V\}$.  We now wish to consider $\frac{\d}{\d s}V(sq)$.  For simplicity of notation, let $V := V(sq)$, $V' := \frac{\d}{\d s}V(sq)$, $R := R(\X+sq-V)$, and $R' := R'(\X+sq-V)$:
\begin{align*}
&V' \EP[\exp(-R)] - V \EP[(q - V')R'\exp(-R)] = \EP[q \exp(-R)] - \EP[sq(q - V')R'\exp(-R)]\\
&\Rightarrow \; V'\EP[(1 - [sq-V]R')\exp(-R)] = \EP[q(1 - [sq-V]R')\exp(-R)]\\
&\Rightarrow \; V' = \frac{\EP[q(1-[sq-V]R')\exp(-R)]}{\EP[(1 - [sq-V]R')\exp(-R)]}.
\end{align*}
That is, $V' = f^q(s)$ in this case.

Third assume $s \in \operatorname{int}\{s \in \R_+ \; | \; sq \not\in \dom V\}$.  By construction, $\bar V(sq) = \essinf [\X + s q]$.  Immediately this implies $\frac{\d}{\d s}\bar V(sq) = f^q(s)$ in this case by Assumption~\ref{ass:essinf}.

Finally, by continuity of $s \mapsto \bar V(sq)$ (see Lemma~\ref{lemma:barV}\eqref{lemma:barV-cont}) and the fact that $\bar V(0)=0$ the result follows.
\item The assumption implies, for $sq \in \dom V$, $(1 - [sq - V(sq)]R'(\X+sq-V(sq)))\exp(-R(\X+sq-V(sq))) \geq 0$ a.s.\ (resp.\ strictly positive if strict monotonicity).  Therefore $f^q(s) \geq \essinf q$ for every $s \in \R_+$ trivially and if $sq \in \dom V$ with $\P(q > \essinf q) > 0$ then $f^q(s) > \essinf q$.
\item The result follows trivially by continuity of $R'$ and $\bar V$.
\item By the relation in Property~\eqref{lemma:f-exist}, if $s \mapsto \bar V(sq)$ is concave then $f^q$ is nonincreasing.  Therefore the relation holds by Lemma~\ref{lemma:barV}\eqref{lemma:barV-concave}.
\end{enumerate}
\end{proof}

\subsection{Proof of Lemma~\ref{lemma:barf}}
\begin{proof}
\begin{enumerate}
\item $\bar f^q(0) = \EP[q\exp(-R(\X))]/\EP[\exp(-R(\X))]\geq \essinf q$ and, for $s > 0$, $\bar f^q(s) = \frac{\bar V(sq)}{s} \geq \essinf q$ by Lemma~\ref{lemma:barV}\eqref{lemma:barV-bound}.  Now consider $\P(q > \essinf q) > 0$:
    \begin{enumerate}
    \item If $sq \in \dom V$ then it follows that $\bar f^q(s) = \frac{\EP[q\exp(-R(\X+sq-V(sq)))]}{\EP[\exp(-R(\X+sq-V(sq)))]} > \essinf q$.
    \item If $sq \not\in \dom V$ then, by Theorem~\ref{thm:exists}, it must be that $\bbd = \Rplus$ and, in particular, $\essinf \X \in \Rplus$.  Therefore, $\bar f^q(s) = \frac{\essinf [\X + s q]}{s} =\frac{\essinf \X}{s} + \essinf q> \essinf q$.
    \end{enumerate}

\item Recall, $\bar f^q(s) = \bar V(sq)/s$ for $s > 0$ and $\bar f^q(0) = \EP[q\exp(-R(\X))]/\EP[\exp(-R(\X))]$.  By continuity of $\bar V$ (see Lemma~\ref{lemma:barV}\eqref{lemma:barV-cont}), continuity of the inverse demand function holds so long as $\lim_{s \to 0} \bar f^q(s) = \EP[q\exp(-R(\X))]/\EP[\exp(-R(\X)]$.  We will consider two cases: $\bbd = \R$ and $\bbd = \Rplus$.
    \begin{enumerate}
    \item Let $\bbd = \R$.  Then $\bar V = V$ by construction as $\dom V = L^\infty$.  Note that, here, $\bar f^q(s) = V(sq)/s = \EP[q \exp(-R(\X+sq-V(sq)))]/\EP[\exp(-R(\X+sq-V(sq)))]$ for every $s > 0$.  Noting $V(0) = 0$ implies $\lim_{s \searrow 0} \bar f^q(s) = \EP[q\exp(-R(\X))]/\EP[\exp(-R(\X))]$ and the result is proven. 
    \item Let $\bbd = \Rplus$.  
    First we want to consider a small remark on the domain of $V$; if $Z \in L^\infty$ such that $\| Z-\essinf Z\|_{\infty} < \essinf\X/2$ then $Z \in \dom V$ since 
$H_Z(v) \leq \esssup Z < \essinf\X/2 + \essinf Z < \essinf[\X+Z]$ for any $v \in [\essinf Z,\essinf[\X+Z])$. 
    
    If $q$ is deterministic, then by Lemma \ref{lemma:barV}\eqref{lemma:barV-bound} $\bar V(sq) =sq$, and therefore $\bar f^q(s)=q$. Otherwise, if $\esssup q - \essinf q >0$, we have that  $\bar f^q(s) = \frac{\bar V(sq)}s = \frac{V(s(q - \essinf q)) + s \essinf q}{s}$ 
    for $s < \essinf\X/[2(\esssup q - \essinf q)]$.  As with the prior case, for this small $s$,
    \[\bar f^q(s) = \frac{\EP[(q-\essinf q)\exp(-R(\X + s[q - \essinf q] - V(s[q - \essinf q])))]}{\EP[\exp(-R(\X + s[q - \essinf q] - V(s[q - \essinf q])))]} + \essinf q.\]
    Again, noting $V(0) = 0$ then trivially $\lim_{s \searrow 0} \bar f^q(s) = \EP[q\exp(-R(\X))]/\EP[\exp(-R(\X))]$.
    \end{enumerate}

\item Note that $\bar f^q(s) = \sup\{p \in \R_+ \; | \; \EP[(q-\essinf q -p)\exp(-R(\X+s(q-\essinf q-p)))] \geq 0, \; \essinf\X - sp \in \bbd\} + \essinf q$ by using the optimization representation of $\bar V$ provided in Proposition~\ref{prop:opt-rep}.  By construction, $\bar f^q(s) = \bar f^{q-\essinf q}(s) + \essinf q$.  Therefore, monotonicity holds in general if it is true for every random variable $q$ such that $\essinf q = 0$.  Assume $\essinf q = 0$ and let $s_1 \geq s_2$.  Fix $p \in [0,\bar f^q(s_1)]$.  Immediately $\essinf\X - s_1 p \leq \essinf\X - s_2 p$; by $\essinf\X - s_1 p \in \cl\bbd$ the same must be true for $\essinf\X - s_2 p$.  Therefore monotonicity follows if $\EP[(q-p)\exp(-R(\X+s_2(q-p)))] \geq \EP[(q-p)\exp(-R(\X+s_1(q-p)))]$.  This holds because
    \[\frac{\d}{\d s}\EP[(q-p)\exp(-R(\X+s(q-p)))] = -\EP[(q-p)^2 R'(\X+s(q-p))\exp(-R(\X+s(q-p)))] \leq 0.\]

\end{enumerate}
\end{proof}

\section{Proofs for Section \ref{sec:cs}}
\subsection{Proof of Proposition \ref{prop:V-exponential}}
\begin{proof}
Recall that $\alpha = \left(\sum_{i = 1}^n \frac{1}{\alpha_i}\right)^{-1}$.  Let $\alpha^* := \left(\sum_{i = 1}^{n+1} \frac{1}{\alpha_i}\right)^{-1} = \left(\frac{1}{\alpha} + \frac{1}{\alpha_{n+1}}\right)^{-1} < \alpha$.  Consider now $\frac{\d}{\d\alpha}V(Z;\alpha)$.  We will show that $\frac{\d}{\d\alpha}V(Z;\alpha) \leq 0$ and therefore, it follows that $V(Z;\alpha) \leq V(Z;\alpha^*)$ for every $Z \in L^\infty$.
\begin{align*}
\frac{\d}{\d\alpha}V(Z;\alpha) &= \frac{\EP[\exp(-\alpha Z)]\frac{\d}{\d\alpha}\EP[Z \exp(-\alpha Z)] - \EP[Z \exp(-\alpha Z)]\frac{\d}{\d\alpha}\EP[\exp(-\alpha Z)]}{\EP[\exp(-\alpha Z)]^2} \\
&= \frac{\EP[\exp(-\alpha Z)]\EP[-Z^2 \exp(-\alpha Z)] - \EP[Z \exp(-\alpha Z)]\EP[-Z \exp(-\alpha Z)]}{\EP[\exp(-\alpha Z)]^2}.
\end{align*}
Let $x = \alpha Z$. Then $\frac{\d}{\d\alpha}V(Z;\alpha) \leq 0$ if, and only if, $\EP[\exp(-x)]\EP[x^2 \exp(-x)] - \EP[x \exp(-x)]^2 \geq 0$.
\begin{align*}
&\EP[\exp(-x)]\EP[x^2 \exp(-x)] - \EP[x \exp(-x)]^2\\
&= \EP[\exp(-x)]\EP[x^2 \exp(-x)] - \EP[\exp(-\frac{x}{2})x\exp(-\frac{x}{2})]^2\\
&\geq \EP[\exp(-x)]\EP[x^2 \exp(-x)] - \EP[\exp(-\frac{x}{2})^2]\EP[(x\exp(-\frac{x}{2}))^2]\\
&= \EP[\exp(-x)]\EP[x^2 \exp(-x)] - \EP[\exp(-x)]\EP[x^2\exp(-x)] = 0,
\end{align*}
where the inequality above follows from the Cauchy-Schwartz inequality.
\end{proof}

\subsection{Proof of Proposition \ref{prop:V-power}}
\begin{proof}
\begin{enumerate}
\item By Theorem~\ref{thm:unique}, $Z \in \dom V$ if and only if $H_Z(\X+\essinf Z) \leq \X + \essinf Z$ provided 
\blue{$\X+\essinf Z \in \dom H_Z$} or $\liminf_{v \nearrow \X+\essinf Z} H_Z(v) < \X + \essinf Z$ if 
\blue{$\X+\essinf Z\not\in\dom H_Z$}.  First assume 
\blue{$\X+\essinf Z \in \dom H_Z$}, then $H_Z(\X+\essinf Z) \leq \X+\essinf Z$  if and only if $\EP[Z(Z-\essinf Z)^{-\eta}] \leq (\X+\essinf Z)\EP[(Z-\essinf Z)^{-\eta}]$.  Rearranging terms completes the proof.  Now assume 
\blue{$\X+\essinf Z\not\in\dom H_Z$}. This can only happen if $\EP[(Z-\essinf Z)^{-\eta}] = \infty$.  However, by $\eta \leq 1$, $\EP[(Z-\essinf Z)^{1-\eta}] \leq \|Z-\essinf Z\|_{\infty}^{1-\eta} < \infty$.  Therefore $\liminf_{v \nearrow \X + \essinf Z} H_Z(v) = \essinf Z < \X + \essinf Z$ and, thus, $Z \in \dom V$.  Notably, in this case $\EP[(Z-\essinf Z)^{1-\eta}] \leq \X\EP[(Z-\essinf Z)^{-\eta}]$ as well.
\item This follows immediately by Lemma~\ref{lemma:barV} as $z\exp(-R(\X+z)) = z(\X+z)^{-\eta}$.
\item This follows by the same logic as Lemma~\ref{lemma:barV} because $z\exp(-R(\X+z)) = z(\X+z)^{-\eta}$.
\end{enumerate}
\end{proof}

\end{document}